\documentclass[11pt]{amsart}
  
\usepackage{fancyheadings}
\usepackage{amsmath}
\usepackage{amsfonts,amsmath,epsfig,color}
\usepackage{placeins} 
\usepackage{flafter}  

\usepackage{tikz}
\usetikzlibrary{trees}

\usepackage{euscript}
\usepackage{multicol}
\usepackage{tabularx}
\usepackage{fancyhdr}
\usepackage{pgfplots}

\usepackage[letterpaper,asymmetric,left=1in,right=1in,top=0.75in,bottom=0.75in,
bindingoffset=0.0in]{geometry}

\newtheorem{theorem}{Theorem}

\newtheorem{lemma}{Lemma}
\newtheorem{definition}{Definition}
\newtheorem*{theoremA}{Theorem A (Ferromagnetic Case)}
\newtheorem*{theoremB}{Theorem B (Antiferromagnetic Case, $q \geq 3$)}
\newtheorem*{theoremC}{Theorem C (Antiferromagnetic Case, $q = 2$)}
\newtheorem*{theoremD}{Theorem D (Renormalization Procedure)}
\newtheorem{remark}{Remark}

\title[Lee-Yang zeros for the Potts Model on the Cayley Tree]{A dynamical approach to studying the Lee-Yang zeros for the Potts Model on the Cayley Tree}

\author{Diyath Pannipitiya}
\address{IU Indianapolis Department of Mathematics\\
         402 North Blackford Street room LD270\\
         Indianapolis, Indiana 46202-3267 \\
         USA}
\email{dinepann@iu.edu}

\author{Roland Roeder}

\address{IU Indianapolis Department of Mathematics\\
         402 North Blackford Street room LD270\\
         Indianapolis, Indiana 46202-3267 \\
         USA}
\email{roederr@iu.edu}

\begin{document}

\begin{abstract}
Let $Z_n(z,t)$ denote
the partition function of the $q$-state Potts Model on the rooted binary Cayley tree of depth~$n$.  Here, $z = {\rm e}^{-h/T}$ and $t = {\rm e}^{-J/T}$ with $h$ denoting an externally applied magnetic field, $T$ the temperature, and $J$ a coupling constant.  One can interpret $z$ as a ``magnetic field-like'' variable and $t$ as a ``temperature-like'' variable.  Physical values $h \in \mathbb{R}$, $T > 0$, and $J \in \mathbb{R}$ correspond to $t \in (0,\infty)$ and $z \in (0,\infty)$. For any fixed $t_0 \in (0,\infty)$ and fixed $n \in \mathbb{N}$ we consider the complex zeros of $Z_n(z,t_0)$ and how they accumulate on the ray $(0,\infty)$ of physical values for $z$ as $n \rightarrow \infty$.   In the ferromagnetic case ($J > 0$ or equivalently $t \in (0,1)$) these Lee-Yang zeros accumulate to at most one point on $(0,\infty)$ which we describe using explicit formulae.   In the antiferromagnetic case $(J < 0$ or equivalently $t \in (1,\infty)$) these Lee-Yang zeros accumulate to at most two points of $(0,\infty)$, which we again describe with explicit formulae.  The same results hold for the unrooted Cayley tree of branching number two.

These results are  proved by adapting a renormalization procedure that was previously used in the case of the Ising model on the Cayley Tree by M\"uller-Hartmann and Zittartz
(1974 and 1977), Barata and Marchetti (1997), and Barata and Goldbaum (2001).   We then use methods from complex dynamics and, more specifically, the active/passive dichotomy for iteration of a marked point, along with detailed analysis of the renormalization mappings, to prove the main results.
\end{abstract}

\maketitle

\markboth{\today}{\today}

\section{Introduction}\label{SEC:INTRO}

This paper concerns the Lee-Yang zeros for the $q \geq 2$ state Potts Model on the binary Cayley Tree.  Because of the recursive nature in which the Cayley Tree is constructed, there is a suitable
renormalization procedure which makes this problem amenable to methods from dynamical systems.  
This paper is intended for readers from statistical physics and also from mathematics (especially dynamical systems), so we will provide considerable background and motivation in Sections \ref{SEC:INTRO} and~\ref{section 3}.

Throughout this paper we will use the convention that the natural numbers are
all integers greater than or equal to $1$, i.e. $\mathbb{N} = \{1,2,\ldots\}$.
Given $k \in \mathbb{N}$ we let $\mathbb{N}_{\geq k} := \{n \in
\mathbb{N} \, : \, n \geq k\}$.

\subsection{Lee-Yang zeros for the Ising Model}\label{SUBSEC:ISING}

Let $(\Gamma_n)_{n=1}^\infty$ be a sequence of graphs, and let $(V_n)_{n=1}^\infty$ and $(E_n)_{n=1}^\infty$ be the corresponding vertex set (sites) and edge set (bonds). Each such graph is interpreted as a finite approximation to a magnetic material and classically one might let $\Gamma_n$ be an $n \times n$ piece of the $\mathbb{Z}^2$ square lattice or an $n \times n \times n$ piece of the $\mathbb{Z}^3$ cubical lattice.

In the Ising model, one magnetic particle (an electron, for example) is at each vertex, and two particles interact if and only if they are connected by an edge.  Each particle is assigned a magnetic moment, called \textit{spin}, which is represented in the model by the discrete variable $\sigma (i)\in \{-1,+1\}$ which describes the spin at vertex $i$. For each spin configuration, $\sigma : V_n\rightarrow \{\pm 1\}$, define the Hamiltonian (energy) of $\sigma$ by
\begin{align}
    H_n(\sigma):=-J\cdot \sum_{<i,j>\in E_n}\sigma(i) \sigma(j)-h\cdot \sum_{i\in V_n}\sigma(i).
\end{align}
    Here $J>0$ is the ferromagnetic coupling constant\index{coupling constant} and $h$ is the strength of the external magnetic field.
    
For a fixed temperature $T>0$, the Gibbs-Boltzmann weight of the spin configuration $\sigma$ is given by \begin{align}
         W_n(\sigma):=e^{-H_n(\sigma)/T}.
     \end{align}
     (We set the Boltzmann constant $k_B = 1$.)  Summing over all the possible spin configurations gives the partition function
     \begin{align}
         Z_n(h,T):=\sum_{\sigma}W_n(\sigma)= \sum_{\sigma} e^{-H_n(\sigma)/T} .\label{(1.3)}
     \end{align}
By the Gibbs-Boltzmann hypothesis,
the probability $P_n(\sigma)$ of the system being in the state $\sigma$ is proportional to $W_n(\sigma)$. Thus, 
 \[P_n(\sigma) = W_n(\sigma)/Z_n(h,T).\]

Since $Z_n(h,T)$ is a finite sum of exponentials, it can never be zero for real values of $h$ and $T$.   However, the zeros of $Z_n(h,T)$ at {\em complex} values of $h$ or $T$ do occur.  The way in which they accumulate
to real values of $h$ and $T$ in the limit as $n$ tends to infinity gives
information about the phase transitions in the model.    This interpretation
dates back to the works of Lee and Yang \cite{Lee-Yang1,Lee-Yang2}.   We will use this interpretation
as a motivation for studying the complex zeros of $Z_n(h,T)$ for the Cayley Tree and their limiting behavior as $n$ tends to infinity, without pursuing more deeply how this relates to phase transitions.

 By letting
 \begin{align}\label{EQN:DEF_Z_T}
     z=e^{-h/T} \ \text{(field-like) and}\ t=e^{-J/T}\ \text{(temperature-like)},
     \end{align}
    we get $Z_n(z,t)$ as a polynomial of $z$ and $t$ when multiplied by $z^{\mid V_n\mid} t^{\mid E_n\mid}$ to clear the denominator.
    For the physical values $T>0$ and $h\in \mathbb{R}$, we must have $t\in (0,1)$ and $z\in (0,\infty)$.    Because the partition function becomes
    a polynomial in the $(z,t)$ variables we will study it exclusively in
    terms of $z$ and $t$.

    The initial studies of phase transitions for the Ising model considered what happens as $T$ is varied with fixed $h=0$.  This corresponds to setting $z=1$ and studying the zeros
    of $Z_n(1,t)$ in the complex $t$ plane and how they accumulate to points
    on the interval $(0,1)$ of physically relevant values of $t$.   Such zeros in the complex $t$ plane are called {\em Fisher zeros} in honor of Michael Fisher~\cite{FISHER,BK}.
    
    It is also interesting to fix $T=T_0>0$ and to vary $h$.  This corresponds to fixing a value of $t = t_0 \in (0,1)$ and studying the
    the zeros of $Z_n(z,t_0)$ in the complex $z$ plane and also how they accumulate to the real ray $(0,\infty)$ of physically relevant values
    of $z$. In 1952, Tsung-Dao Lee \index{Lee, Tsung-Dao} and Chen-Ning Yang \index{Yang, Chen-Ning} published two important papers in statistical mechanics and proved a series of Lee-Yang theorems \cite{Lee-Yang1,Lee-Yang2}. The most interesting theorem is the following: 
\label{LY1} \begin{theorem}[Lee-Yang Circle Theorem]
     For $t\in[0, 1]$, the complex zeros in $z$ of the partition function $Z(z, t)$ of the Ising model on any graph lie on the unit circle $\mathbb{T} = \{z\in \mathbb{C}:\ \mid z\mid = 1\}$.
 \end{theorem}
 \noindent
Interpreting parameters $z_0 \in (0,\infty)$ at which zeros of $Z_n(t_0,z)$ accumulate when $n \rightarrow \infty$ as corresponding to phase transitions, 
this implies that for any fixed $t_0\in [0,1]$, the only physical parameter at which we could potentially observe a phase transition is when $z_0=1$ (equivalently when $h_0=0$).   Because of this impressive theorem, the zeros of $Z_n(z,t_0)$ in the complex $z$ plane are called {\em Lee-Yang} zeros.

\subsection{Potts Model with an external
magnetic field and its Lee-Yang zeros.}\label{SUBSEC:POTTS}
 Much of the discussion of the previous section carries over directly to the Potts model. Let $q>2$ be a natural number. Now one can consider configurations $\sigma: V_n \longrightarrow \{0, \cdots, q-1\}$. If $J > 0$ is the coupling constant, then the energy of such a configuration when exposed to external magnetic field $h \in \mathbb{R}$ is defined to be 
\begin{align}\label{(2.1)}
    H_n(\sigma)= -J\sum_{(i,j)\in E_n}\delta(\sigma(i),\sigma(j)) - h\sum_{i\in V_n}\delta(\sigma(i),0).
\end{align}
\noindent   See, for example, Equation $20$ in \cite{WuPercolationAndThePottsModel}. 
(Here, $\delta(i,j) = 1$ if $i=j$ and $\delta(i,j) = 0$ otherwise.)  
The partition function $Z_n(z,t)$ and probability $P(\sigma)$ with which a spin configuration occurs are defined exactly in the same way as in the previous section, except that the new Formula (\ref{(2.1)}) for the energy  $H_n(\sigma)$ is used instead of the one in the Ising model. 

Like for the Ising Model, it is convenient to express the partition function
in terms of the ``field-like'' and ``temperature-like'' variables $z$ and $t$ that were defined in (\ref{EQN:DEF_Z_T}).
In the context of the Potts Model, one again 
 calls the zeros of $Z_n(1,t)$ in the complex $t$ plane the {\em Fisher zeros} and, for fixed $t_0 \in [0,1]$, the zeros of $Z_n(z,t_0)$ in the complex $z$ plane, the {\em Lee-Yang zeros}.

 Remark that for $t=0$ and $t=1$ it is easy to compute the partition function.  For any connected graph $\Gamma$ with $k$ vertices one finds
 \begin{align}\label{EQN:PARTITION_FXN_WHEN_T_EQUALS_01}
Z_\Gamma(z,0) = z^{-k} + (q-1) \qquad \mbox{and} \qquad Z_\Gamma(z,1) = \left(q-1+z^{-1} \right)^k.
 \end{align}
 In particular the Lee-Yang zeros when $t=0$ are the $k$-th roots of $1/(1-q)$ and the Lee-Yang zeros
 when $t=1$ all $k$ Lee-Yang zeros are equal to $1/(1-q)$.
 When $t \in (0,1)$ the Lee-Yang zeros are much more difficult to compute and they depend on the graph 
 $\Gamma$ and the temperature-like variable $t$ in a non-trivial way.

Because of its complexity, the $q$-state Potts model with the presence of an external magnetic field has been researched by just a few authors, as described in Section \ref{Related Works}.

\subsection{The antiferromagnetic case}
As in Sections \ref{SUBSEC:ISING} and \ref{SUBSEC:POTTS} above, it is customary to choose the coupling constant $J$ to be positive, making it energetically favorable for spins at neighboring vertices to be aligned.  It corresponds to the {\em ferromagnetic} materials.   However, there are some physical systems for which the opposite phenomenon holds.   They are called {\em antiferromagnetic} and for such systems one assumes that $J < 0$.

\subsection{Binary Cayley Tree}
The {\em $n^{th}$-level rooted binary Cayley tree} is a tree (a simple,
undirected, connected, and finite graph in which any two vertices are connected
by exactly one path) in which one vertex, called the root, is of degree two
with all leaves (vertices of degree one) at a distance (minimum number of edges
to connect) $n$ from the root, and all the other vertices are of degree three.
The {\em $n^{th}$-level unrooted binary Cayley tree} is a tree in which all
vertices have degree $3$ or $1$ and for which there exists a unique vertex of
degree $3$ such that all of the leaves are of distance $n$ from that specified
vertex.  In the remainder of this paper we will denote the rooted binary Cayley
Tree of level $n$ by $\Gamma_n$ and the unrooted binary Cayley Tree of level
$n$ by $\hat{\Gamma}_n$.   We will consider the $q$-state Potts Model on these
two families of graphs, as $n$ approaches infinity, for the rest of this
paper.

\begin{figure}
\begin{center}
\includegraphics[scale=0.38]{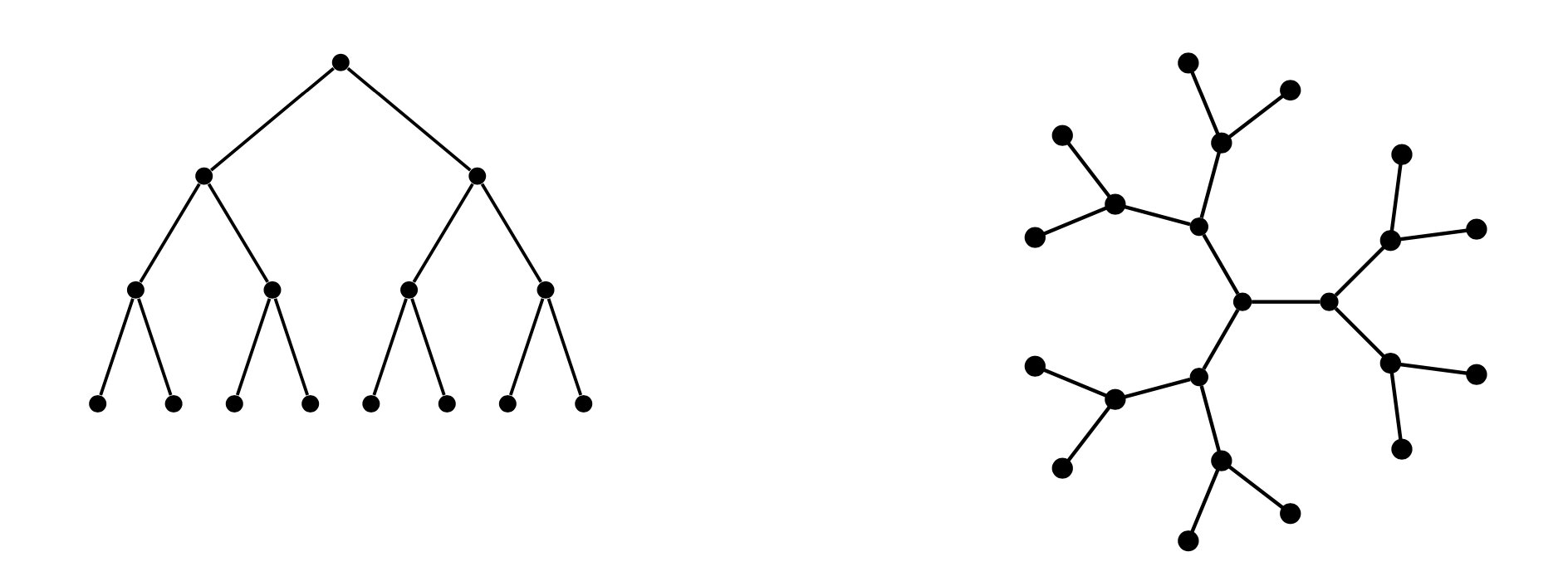}
\end{center}
\caption{Left: rooted binary Cayley Tree $\Gamma_4$ of depth 4.   Right: unrooted binary Cayley tree $\hat{\Gamma}_4$ of depth $4$.}
\end{figure}

\begin{remark}
It is important to note that
many classical treatments of the Ising Model on the Cayley Tree consider the
thermodynamical properties associated only with vertices ``deep'' in the lattice (e.g. the root vertex); see
\cite[Ch. 4]{BaxterBook} and the references therein. The term ``Bethe Lattice'' is typically used to describe such
considerations. Instead, we treat all vertices equally, studying the ``bulk'' behavior of the lattice,
and thus we follow the standard convention of referring to our work as being on the Cayley Tree.
\end{remark}

\subsection{Plots of some Lee-Yang zeros for the Ising and Potts models on the Cayley tree.}\label{section 1.7}

There is a convenient recursive formula that allows one to compute the Lee-Yang zeros of the $n^{th}$ rooted Cayley tree for the relatively small values of $n$. See Theorem D in Section \ref{main results} below. It applies to both the Ising and Potts models. Using that theorem, we have generated figures plotting Lee-Yang zeros for the Ising model and $3$-state Potts model.

The left side of Figure \ref{Fig 1.1} shows the Lee-Yang zeros for the Ising model on the binary Cayley tree when $n=5$ and when $t=0.0625$. Note that the zeros appear to lie perfectly on the unit circle $\mid z \mid =1$ as described by the Lee-Yang circle theorem. Moreover, there exists a critical temperature $t_{\rm crit}=1/3>0$ such that when $t \leq t_{\rm crit}$ the Lee-Yang zeros accumulate to $1$ on the positive real axis and such that when $t>t_{\rm crit}$, the Lee-Yang zeros stay away from the positive real axis (right side of Figure \ref{Fig 1.1}).
(For a rigorous justification of this phenomenon, see \cite[Theorem A]{Limiting}.)

The situation is dramatically different for the $3$-state Potts model on the binary Cayley tree because the Lee-Yang zeros no longer lie on the unit circle. In fact, the Lee-Yang zeros of the $5^{th}$ rooted binary Cayley tree for the $3$-state Potts model seem to lie inside the unit circle (left side of Figure \ref{Fig 1.3}). However, like for the Ising Model, there is a critical temperature $t_{\rm crit}= \frac{1+\sqrt{73}}{36}\approx 0.265$ such that for any $t
\leq t_{\rm crit}$, these zeros accumulate to a point on the positive real axis and such that when $t>t_{\rm crit}$, the Lee-Yang zeros stay away from the positive real axis (right side of Figure \ref{Fig 1.3}).  This will be rigorously justified in Theorem A.  (The reader may
also compare this description with a conjectural description of the accumulation locus of the Lee-Yang zeros described in Remark \ref{REM:CONJ_DESCR}.)

\begin{figure}
    \centering
   \includegraphics[width=0.35\linewidth]{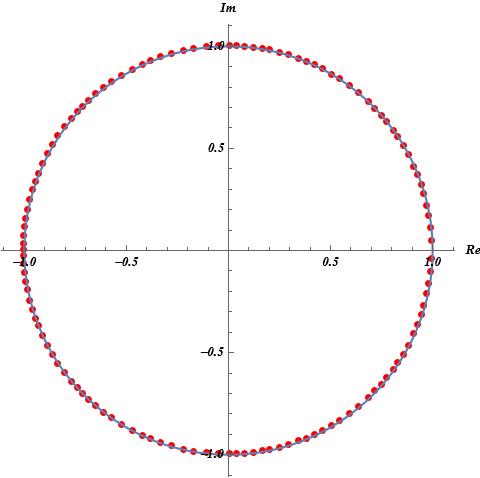}\hspace{0.3in}\includegraphics[width=0.35\linewidth]{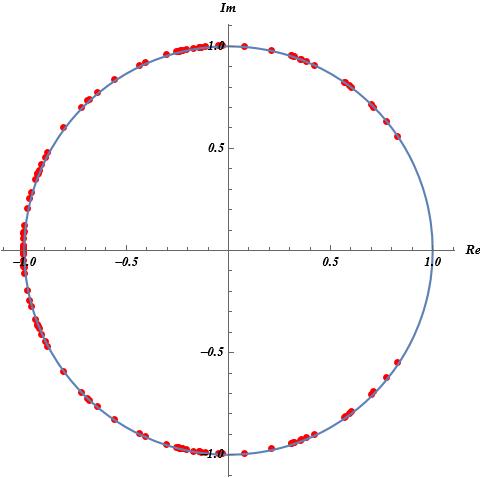}
   \caption{The Lee-Yang zeros for the $5^{th}$ rooted binary Cayley Tree for the Ising model.  Left: $\displaystyle t=0.0625<t_{\rm crit}=1/3$.  Right: $t=0.5>t_{\rm crit}=1/3$.}
    \label{Fig 1.1}
 \end{figure}

\begin{figure}
    \centering
  \includegraphics[width=0.35\linewidth]{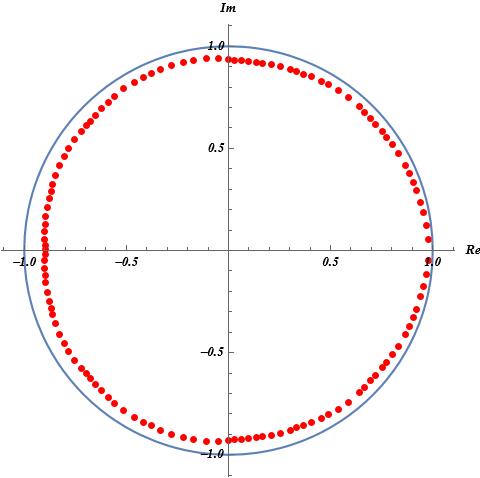}\hspace{0.3in}\includegraphics[width=0.35\linewidth]{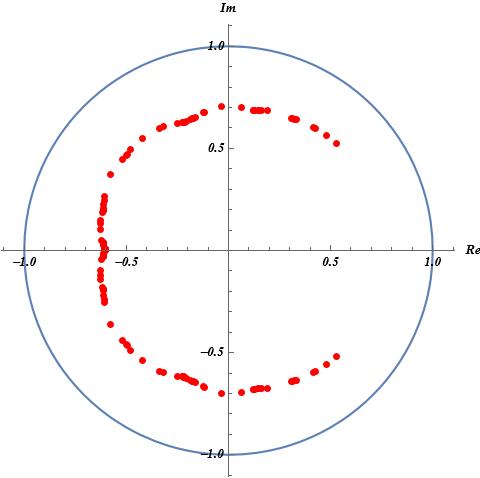}
    \caption{The Lee-Yang zeros for the $5^{th}$ rooted binary Cayley Tree for the $3$-state Potts model.  Left: $t=0.0625<t_{\rm crit}= \frac{1+\sqrt{73}}{36}\approx 0.265$.  Right: $t=0.5>t_{\rm crit}= \frac{1+\sqrt{73}}{36}\approx 0.265$.}
    \label{Fig 1.3} 
\end{figure}

\FloatBarrier

\subsection{Main Results}\label{main results}\label{section 1.8}
For the remainder of the paper we will denote by $Z_n(z,t,q)$ the partition function for the $q$ state Potts model on the rooted binary Cayley tree of depth $n$.    The partition function for the unrooted binary Cayley tree of depth $n$ will be denoted by $\hat{Z}_n(z,t,q)$.   Often we will drop the dependence of $q$ from the notation when it is clear: $Z_n(z,t) \equiv Z_n(z,t,q)$ and $\hat{Z}_n(z,t) \equiv \hat{Z}_n(z,t,q)$.

For fixed choices of $t \in [0,\infty]$ and $q \in \mathbb{N}_{\geq 2}$
our results consider the sets
\begin{align}
\mathcal{B}(t,q) &:= \overline{\{z \in \mathbb{C} \, : \,  Z_n(z,t,q) = 0 \, \mbox{for some $n \in \mathbb{N}$}\}}, \quad \mbox{and} \label{DEF_SETS_B} \\
\hat{\mathcal{B}}(t,q) &:= \overline{\{z \in \mathbb{C} \, : \,  \hat{Z}_n(z,t,q) = 0 \, \mbox{for some $n \in \mathbb{N}$}\}}. \nonumber
\end{align}
 Following the Lee-Yang approach to studying phase transitions (see Section \ref{SUBSEC:ISING}) we will be interested in where the sets $\mathcal{B}(t,q)$ and $\hat{\mathcal{B}}(t,q)$ intersect the ray $z \in (0,\infty)$ for various fixed choices
of $t \geq 0$ and $q\in \mathbb{N}_{\geq 2}$.

We first consider the ferromagnetic case $J > 0$ in which physical values of $z$ and $t$ correspond
to $z \in (0,\infty)$ and $0< t <1$.  Let 
\begin{align}\label{EQN:DEF_T1_T2}
    t_{1}(q) = \frac{1}{1+q} \qquad \mbox{and} \qquad t_{2}(q) = \frac{q-2 + \sqrt{q^2+32 q-32}}{18 (q-1)}.
\end{align}
Note that $t_{1}(2) = t_{2}(2) = 1/3$ but that $t_1(q) < t_2(q)$ for all $q \geq 3$.

\begin{theoremA}
For any $t \in [0,1]$, as $n \rightarrow \infty$ 
the Lee-Yang zeros for the $q \geq 2$ state Potts Model on the {\rm (}rooted or unrooted{\rm )} binary Cayley Tree accumulate to $z \in (0,\infty)$ if and only if $t \in [0,t_{2}(q)]$. They do so at a single point 
\begin{align*}
  z_c(t,q) =
	\begin{cases}
		  1\ & \mbox{if }\ 0\leq t\leq t_{1}(q) \\
  \mathcal{Z}_q(t)\ & \mbox{if }\ t_{1}(q) <t\leq t_{2}(q)
\end{cases}
\end{align*}
with
\begin{align}\label{EQN:CALZ_Q}
\mathcal{Z}_q(t) &:= \frac{\left(\substack{
(-27 (q-1)^2 t^4+18 \left(q^2-3 q+2\right) t^3+\left(q^2+14 q-14\right) t^2 +2 (q-2) t+1))
 \\ 
 - \sqrt{(t-1) ((q-1) t+1) \left(9 (q-1) t^2-(q-2) t-1\right)^3}
}\right)}
     {8 t ((q-2) t+1)^3}.
 \end{align}
\end{theoremA} 

\begin{remark}
    Theorem A can be reformulated as saying that for any $q \in \mathbb{N}_{\geq 2}$ we have that $\mathcal{B}(t,q) \cap (0,\infty)$ consists of a single point $z_c(t,q)$ for $0 \leq t \leq t_2(q)$ and that $\mathcal{B}(t,q) \cap (0,\infty) = \emptyset$ for $t_2(q) < t \leq 1$.   The same holds when $\mathcal{B}(t,q)$ (associated with the rooted Cayley Tree) is replaced with $\hat{\mathcal{B}}(t,q)$ (associated with the unrooted Cayley Tree).
\end{remark}

\begin{remark}
In the case of the Ising Model we have $t_1(2) = t_2(2) = 1/3$ so that for $t \in [0,1]$ the only $z \in (0,\infty)$ at which the Lee-Yang zeros can accumulate is $z = 1$.  This is consistent with the Lee-Yang circle theorem. However, when $q \geq 3$ an interesting new phenomenon occurs for the Potts Model.   For the non-empty
interval $t \in (t_1(q),t_2(q)]$ the Lee-Yang zeros accumulate at $z_c(t,q) = \mathcal{Z}_q(t) < 1$.    See Figure \ref{graph of $y=Z(t)$} for the case when $q=3$.
\end{remark}

 \begin{figure}[h!]
    \centering
    \includegraphics[width=0.8\linewidth]{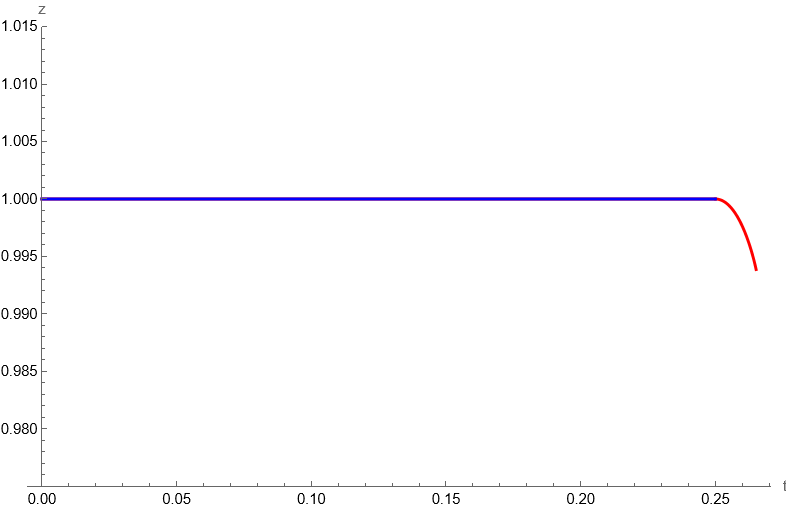}
    \caption{Graph of $z=z_c(t,3)$ (Theorem A).  For $0 \leq t \leq 1/4$ we have
    $z_c(t,3) = 1$ (blue) and for $1/4 < t \leq t_2(q) \approx 0.2651$ we have $z_c(t,3) = \mathcal{Z}_q(t) < 1$ (red).}
    \label{graph of $y=Z(t)$}
\end{figure}

\vspace{0.1in}

We now consider the antiferromagnetic case $J < 0$ in which physical values of $z$ and $t$ correspond
to $z \in (0,\infty)$ and $t > 1$.  Let
\begin{align}\label{EQN:T3}
t_3(q) = \frac{3 \left(3 q-6 + \sqrt{9 q^2-32 q+32}\right)}{2 (q-1)}.
\end{align}

\vbox{
\begin{theoremB}
For any $t > 1$,
as $n \rightarrow \infty$ the Lee-Yang zeros for the $q \geq 3$ state Potts Model on the {\rm (}rooted or unrooted{\rm )} binary Cayley tree accumulate to $z \in (0,\infty)$ if and only if 
  $t\geq t_3(q)$ and $z=z_c^{\pm}(t,q)$  where

\begin{align}\label{EQN:Z_C_PM}
z_c^{\pm}(t,q)  :=
\frac{\left(\substack{
(-3 - 6  (-2 + q)  t - 3  (2 + (-2 + q)  q)  t^2 -
   6  (-2 + q)  (-1 + q)  t^3 + (-1 + q)^2  t^4 )
 \\
 \pm \sqrt{(-1 + t)^3  (1 - t + q t)^3  (-9 + 18  t - 9  q  t - t^2 +
       q  t^2)}
}\right)}
     {8  t  (1 - 2 t + q t)^3}.
 \end{align}

\end{theoremB}
}

\noindent
A plot of $z_c^\pm(t,3)$ versus $t > 1$ is shown in Figure \ref{Fig 1.6}.  


\begin{remark}
    Theorem B can be reformulated as saying that for any $q \in \mathbb{N}_{\geq 2}$ we have that $\mathcal{B}(t,q) \cap (0,\infty) = \emptyset$ for $1 \leq t < t_3(q)$ and that $\mathcal{B}(t,q) \cap (0,\infty)$ consists of the (one or) two points $z_c^{\pm}(t,q)$ for $t_3(q) \leq t < \infty$.   The same holds when $\mathcal{B}(t,q)$ (associated with the rooted Cayley Tree) is replaced with $\hat{\mathcal{B}}(t,q)$ (associated with the unrooted Cayley Tree).
\end{remark}

\begin{figure}[h]
    \centering
    \includegraphics[width=0.8\linewidth]{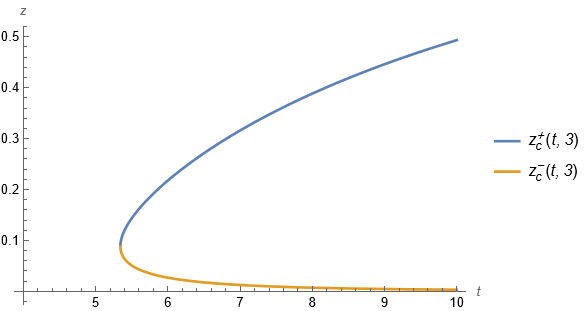}
    \caption{Graph of $z_c^-(t,3)$ and $z_c^+(t,3)$ ($Theorem\ B$).}
    \label{Fig 1.6}
\end{figure}

The antiferromagnetic case for the Ising Model ($q=2$) is slightly different than in Theorem B because one also has Lee-Yang zeros accumulating to $z=1$ when $t \geq t_3(2) = 3$.  We record this fact here:

\begin{theoremC}
For any $t > 1$,
as $n \rightarrow \infty$ the Lee-Yang zeros for the Ising Model ($q = 2$) on the {\rm (}rooted or unrooted{\rm )} binary Cayley tree accumulate to $z \in (0,\infty)$ if and only if 
  $t\geq t_3(q)$ and $z=1$ or $z=z_c^{\pm}(t,2)$  where
\begin{align}\label{EQN:Z_C_PM_QIS2}
z_c^{\pm}(t,2) &:= \frac{t^4-6t^2-3 \pm \sqrt{(t^2-1)^3 (t^2-9)}}{8t}.
 \end{align}
\end{theoremC}

\noindent
A plot of $z_\pm(t,2)$ and also $z=1$ versus $t > 1$ that illustrates Theorem C is given in
Figure \ref{fig:ZCPM_QIS2}.

\begin{figure}
    \centering
    \includegraphics[width=0.8\linewidth]{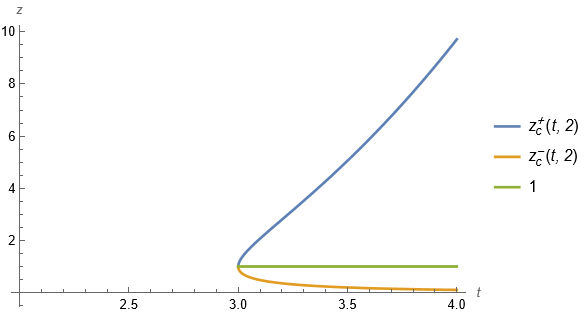}
    \caption{Graphs of $z^-_c(t,2)$ and $z^+_c(t,2)$ (Theorem $C$).}
    \label{fig:ZCPM_QIS2}
\end{figure}

The following theorem is what enables us to use the methods from dynamical systems to prove Theorems A and B above. It is a generalization to the Potts model of a result used by M\"uller-Hartmann-Zittartz \cite{MHZ}, Barata-Marchetti \cite{BarataAndMarchetti}, Chio-He-Ji-Roeder \cite{Limiting}, and others for the Ising model on the Cayley tree.

\begin{theoremD}For $q\in \mathbb{N}_{\geq 2}$, $t\in \mathbb{R}$, and $z,w\in \mathbb{C}$ define $R_{z,t,q}(w)$ and $\hat{R}_{z,t,q}(w)$ by 
\begin{align}\label{EQN:RENORM_MAP}    
R_{z,t,q}(w):=z\left[\frac{t+w+(q-2)tw}{1+(q-1)tw}\right]^2 \qquad \mbox{and} \qquad \hat{R}_{z,t,q}(w):= z\left[\frac{t+w+(q-2)tw}{1+(q-1)tw}\right]^3.
\end{align}
Then:
\begin{itemize}
    \item[(i)] the Lee-Yang zeros of the $q$-state Potts model on the $n^{th}$ level rooted binary Cayley tree are the solutions $z$ to 
    \begin{align*}
        R^n_{z,t,q}(z)=\frac{1}{1-q}, \qquad \mbox{and}
    \end{align*}
    \item[(ii)]the Lee-Yang zeros of the $q$-state Potts model on the $n^{th}$ level unrooted binary Cayley tree are the solutions $z$ to 
    \begin{align*}
        \left(\hat{R}_{z,t,q} \circ R^{(n-1)}_{z,t,q}\right)(z)=\frac{1}{1-q}.
    \end{align*}
\end{itemize}
\end{theoremD}

\begin{remark} The expression $R^n_{z,t,q}(z)$ means that one first iterates $R_{z,t,q}(w)$ $n$-times with respect to $w$ and then substitutes $w=z$.  The expression in Claim (ii) is interpreted analogously.
\end{remark}

\subsection{Comparison of temperatures $t_1(q)$ and $t_2(q)$ to the works of Wang and Wu and of  M\"uller-Hartmann–Zittartz.}

In their 1976 paper \cite{WW}, Wang and Wu considered the ferromagnetic $q$-state Potts Model on the Cayley tree and they computed the temperature $T_c(q) > 0$ such that the zero magnetic field susceptibility diverges for $T < T_c(q)$ but converges for $T > T_c(q)$.    (See also \cite[Section I.E]{WU_review}.) As they explain, this proves that the limiting free energy $f(T,h)$ is non-analytic as $h$ is varied about $0$, at least in the range $0 < T < T_c(q)$.   Expressed in our notation, and specialized to the binary Cayley tree, their result is
\begin{align*}
\frac{J}{T_c(q)} = {\rm ln}\left(\frac{q+\sqrt{2}-1}{\sqrt{2}-1}\right) \qquad \mbox{corresponding to} \qquad t_c = {\rm e}^{-J/T_c} = \frac{\sqrt{2}-1}{q+\sqrt{2}-1}.
\end{align*}
The reader can check that $0 < t_c(q) < t_1(q)$ for all $q \in \mathbb{N}_{q \geq 2}$.   In fact, $t_c(q)$ is the unique value of $t \in (0,1)$ at which $R_{1,t,q}'(1) = \sqrt{2}$.   The reason for this connection between $t_c(q)$ and the derivative of the renormalization map $R_{z,t,q}$ (in the context of the Ising model) is nicely explained in the papers  of M\"uller-Hartmann–Zittartz \cite{MHZ} and of M\"uller-Hartmann \cite{MH}.

Wang and Wu also explain that the higher derivatives of $f(T,h)$ with respect to $h$ diverge at $h=0$ over the wider
range of temperatures $0 < T < T_{\rm BP}(q)$, where $T_{\rm BP}(q)$ is the ``Bethe-Peierls'' temperature.   Thus it is more interesting to compare the Bethe-Peierls temperature $t_{\rm BP}(q) = {\rm e}^{-J/T_{\rm BP}(q)}$ with $t_1(q)$ and $t_2(q)$.   As explained by the papers of M\"uller-Hartmann–Zittartz and of M\"uller-Hartmann in the context of the Ising Model, this corresponds precisely to the unique temperature $t_{\rm BP}(q)$ at which $R_{1,t,q}'(1) = 1$.   A simple calculation shows that $t_{\rm BP}(q) = t_1(q) = \frac{1}{1+q}$.   Indeed, this condition shows up naturally in our proof of Theorem A.

\subsection{Related Works}\label{Related Works}
An early study of the Lee-Yang zeros for the $q \geq 3$ state Potts Model was done numerically by Kim-Creswich \cite{KC} in 1998. Their observations included that the Lee-Yang zeros no longer lie on the unit circle $|z| = 1$ when $q \geq 3$.   Much closer to our paper is the 2002 paper of Myshlyavtsev-Ananikian-Sloot \cite{GAS} where the Lee-Yang zeros for the $q$-state Potts model on the Cayley Tree are studied using a renormalization mapping similar to the one used here, combined with physical reasoning and also numerical experiments.   A main novelty of our paper that goes beyond this nice work of Ghulghazaryan-Ananikian-Sloot is the use of the active/passive dichotomy from complex dynamics; see Section \ref{section 1.9} below.

Other studies regarding the Potts model with non-zero external magnetic field were done by Chang and Shrock in \cite{Chang2009} and \cite{Chang2015}, Shrock and Xu in \cite{Shrock_2010} and \cite{Shrock_2010_2}, McDonald and Moffat \cite{McDonald_Moffat}. 

Studies of the Lee-Yang zeros for the Ising Model are more common, so we mention the ones that are closest to the present work.   Studies of the Lee-Yang zeros for the Ising Model on the Cayley Tree date back to M\"uller-Hartmann \cite{MH}, M\"uller-Hartmann and Zittartz \cite{MHZ}, whose works were followed by that of Barata and Marchetti \cite{BarataAndMarchetti}, Barata and Goldbaum \cite{BarataAndGoldbaum}. Some of the most recent work was done by Chio, He, Ji, and Roeder \cite{Limiting}. 
Extensive studies regarding the Lee-Yang zeros for the Ising Model on the Diamond Hierarchical  Lattice were done by Bleher, Lyubich, and Roeder \cite{BLR}. Studies of Locations of Lee-Yang zeros can be found in the works of Bencs, Buys, Guerini, and Peters \cite{BENCS_BUYS_GUERINI_PETERS_2022}, Regts and Peters \cite{PetersAndGuus}, Camia, Jiang, and Newman \cite{Newman1}, Hou, Jiang, and Newman \cite{Newman2}, and the references therein. Studies of the limiting measure of Lee-Yang zeros for the Curie-Weiss Model were done by Kabluchko \cite{kabluchko2022}.

Note that studying statistical physics on the Cayley tree falls within the wider context of statistical physics on hierarchical lattices. For a sample of recent works see: Akin and Berker \cite{AkinBerker}, De Simoi and Marmi \cite{DEMA}, De Simoi \cite{DE}, Jiang, Qiao, and Lan \cite{JiangQiaoLan}, Myshlyavtsev, Myshlyavtsev, and Akimenko \cite{MYSHLYAVTSEV}, Alvarez \cite{alvarez}, Chio and Roeder \cite{chio}, Chang, Roeder, and Shrock \cite{Chang2015}, and Bleher, Lyubich, and Roeder \cite{bleher2020}.

\subsection{Structure of the paper.}
In Section \ref{section 3} we present some tools from complex dynamics that
will be needed to prove Theorems A-C.   Section \ref{SECTION_PROOF_OF_THMA}
is devoted to the proof of Theorem A and Section \ref{SECTION_PROOF_OF_THMB} is
devoted to the proof of Theorem B.   It is followed by a short Section
\ref{SEC:PROOF_THMC} which sketches the proof of Theorem C.   A technical lemma
that is used in the proofs of Theorem A is presented in Appendix
\ref{APPENDIX_TECHNICAL_LEMMAS}. Even though Theorem D is used in the proofs of
Theorems A-C, we have relegated its proof to the Appendix
\ref{APPENDIX_THMC} because it is relatively standard.   In particular, it is
similar to the derivation of the renormalization mapping for the Ising Model on
the Cayley tree that was presented in \cite{Limiting}. 

\subsection*{Acknowledgments}
We are very grateful to Arnaud Cheritat for making the computer images of the active locus shown in Figures \ref{AC1} through \ref{t=6b} and also for describing his algorithm to us.   We also thank Suzanne and Brian Boyd for their help using the Dynamics Explorer computer software to help us discover the statements proved in this paper.  We are very grateful to Pavel Bleher, Bruce Kitchens, and Robert Shrock for their many helpful suggestions.   We thank Thomas Gauthier for his comments about the active/passive dichotomy and we thank Krishna Chaitanya Kalidindi, Andres Quintero, and Yapeng Zhao for their comments on the writing.   Both authors were supported by NSF grant DMS-2154414.

\section{Background from complex dynamics and initial consequences for our problem.}\index{Complex Dynamics}\label{section 3}
The proofs of Theorem A-C are based on the recursive formula in Theorem D and some ideas in complex dynamics, which we present in
Section \ref{section 1.9}, below.   Section~\ref{SUBSEC:APPLICATIONS_TO_RG_MAP} concerns application of the methods from Section \ref{section 1.9} to the Renormalization mapping $R_{z,t,q}(w)$.
The ideas presented in Sections \ref{section 1.9} and \ref{SUBSEC:APPLICATIONS_TO_RG_MAP} allow us to plot high-quality computer-generated images of the accumulation set of Lee-Yang zeros of the Potts model on the $n^{th}$-level binary Cayley tree as $n\rightarrow \infty$. This will be done in Section \ref{section 1.10}.

\subsection{Active-Passive Dichotomy in Complex Dynamics.}\label{section 1.9}
 We refer the reader to \cite{MILNOR_BOOK} for general background on holomorphic dynamics and to \cite[Section V]{ChangRoederShrock} for a similar discussion of the topics presented below that is also written for an audience from statistical physics.  

Let $\Lambda \subset \mathbb{C}$ be open and let $\mathbb{C}_{\infty}$ be the Riemann sphere \index{Riemann sphere}. Let $(f_\lambda)_{\lambda \in \Lambda}$ be a family of rational maps \index{rational map} from the Riemann sphere to itself that depends holomorphically on the parameter $\lambda$.  Let $a(\lambda)$ be a choice of initial conditions for the iterates of $f_\lambda$ which also depends holomorphically on $\lambda$. It is called a {\em marked point}\index{Marked Point} for $f_\lambda$.

\begin{definition}\label{def 3.1.1} A marked point $a(\lambda)$ is called {\em passive}\index{passive parameter} for $f_\lambda$ at the parameter $\lambda_0$ if the sequence $(g_n)_{n=1}^\infty$ of functions defined by $g_n(\lambda):=f^n_\lambda(a(\lambda))$ forms a normal family\index{normal family} in some neighborhood of~$\lambda_0$. Else, we say $a(\lambda)$ is {\em active}\index{active parameter} for $f_\lambda$ at the parameter $\lambda_0$.
\end{definition}

\noindent
See \cite[Chapter VII]{CONWAY_BOOK1} for the definition of normal family.

\begin{remark}
Historically, the above definitions were applied in the case where the marked point $a(\lambda)$ was a critical point of $f_\lambda(z)$ for each $\lambda \in \Lambda$, to study bifurcations in the dynamics of $f_\lambda$.   See, for example, \cite{MCM}.  However, recently there has been considerable interest in the dynamical properties of noncritical marked points, with some of the motivations coming from problems in arithmetic dynamics.
As a sample of such recent papers, we refer the reader to \cite{DEMARCO1,DEMARCO2,GAUTHIER1,FG_BOOK} and the references therein.
\end{remark}

 The set of all passive parameters for the marked point $a(\lambda)$ is an open set. It is called the \textit{passive locus for $a(\lambda)$}\index{passive locus}. Similarly, the set of active parameters for the marked point $a(\lambda)$ is called the \textit{active locus for $a(\lambda)$}\index{active locus}.

\begin{lemma}\label{prop 3.1.1}
    Let $\ell \geq 2$ be a natural number.   Then, marked point $a(\lambda)$ is passive for $f_\lambda(z)$ at $\lambda = \lambda_0$ if and only if it is passive for the $\ell$-th iterate $f^\ell_\lambda(z)$ at $\lambda = \lambda_0$.
\end{lemma}
\begin{proof}
    Suppose that the marked point $a(\lambda)$ is passive for $f^\ell_\lambda(z)$ at $\lambda_0$. To show that $a(\lambda)$ is passive for $f_\lambda(z)$ at $\lambda_0$, consider any subsequence $(n_k)_k$ of the sequence $(n)$ of all natural numbers.  There must then be some $m \in \{0,\ldots,\ell-1\}$
    such that for infinitely many indices $k$ we have $n_k \equiv m ({\rm mod} \ \ell)$.
    We denote this subsequence by $(n_{k_j})_j$ and for each $j$ we write $n_{k_j} = \ell \, q_{j} + m$ with $(q_j)_j$ an increasing sequence of natural numbers.   Since $a(\lambda)$ is passive for $f^\ell_\lambda(z)$ we can find an open neighborhood $U$ of $\lambda_0$ and a further subsequence $(q_{j_p})_p$ of $(q_j)_j$ such that
    $\left(f^{\ell} \right)^{q_{j_p}}(a(\lambda))$ converges uniformly on compact subsets of $U$ to some holomorphic function $g(\lambda)$.   Then, $f^{\ell \cdot {q_{j_p}}+m}(a(\lambda))$ converges uniformly on compact
    subsets of $U$ to $f^m \circ g(\lambda)$.   Therefore, given the original subsequence $(n_k)_k$ of $(n)$ we have found a further subsequence $(\ell \cdot  q_{j_{p}} + m)_p$ along which $f^{\ell \cdot {q_{j_p}}+m}(a(\lambda))$ converges uniformly on compact subsets of $U$.   Therefore, the family of
    functions $\{f^n(a(\lambda))\}_{n=1}^\infty$ is a normal family on $U$ implying that $\lambda_0 \in U$ is a passive parameter for $f_\lambda(z)$.

 The converse implication immediately follows from the definition of normal family.
\end{proof}

\begin{definition}
    Let $z_{rep}(\lambda_0)$ be a repelling fixed point of $f_{\lambda_0}$. Because $f_\lambda$ varies holomorphically with $\lambda$, we can find a holomorphic map $z_{rep}(\lambda)$ defined in a neighborhood $U$ of $\lambda_0$ such that $z_{rep}(\lambda)$ is a repelling fixed point of $f_\lambda$ for all $\lambda \in U$. Let $n_0\in \mathbb{N} \cup \{0\}$. We say $f_{\lambda_0}^{n_0}$ maps $a(\lambda_0)$ \index{non-persistently} {\em non-persistently} onto the repelling  fixed point $z_{rep}(\lambda_0)$ if
    \begin{align}
        f_{\lambda_0}^{n_0}(a(\lambda_0))=z_{rep}(\lambda_0)\ \text{and}\ f^{n_0}_\lambda(a(\lambda)) \not\equiv z_{rep}(\lambda)\ \text{on}\ U.
    \end{align}
(Here, $f_\lambda^0$ is interpretted as the identity map.)
The definition that $f_{\lambda_0}$ maps the marked point $a(\lambda_0)$ non-persistently onto a repelling periodic point of period $p > 1$ 
is defined analogously by replacing $f_{\lambda_0}$ with $f_{\lambda_0}^p$.
\end{definition}

\begin{lemma} \label{Lemma 3.1.1} Suppose $f_{\lambda_0}^{n_0}$ maps $a(\lambda_0)$ non-persistently onto a repelling periodic point for $f_{\lambda_0}(z)$. Then $\lambda_0$ is an active parameter for the marked point $a(\lambda)$ under $f_\lambda$. 

 If $a(\lambda_0)$ is in the basin of attraction of an attracting periodic point for $f_{\lambda_0}$, then $\lambda_0$ is a passive parameter for the marked point $a(\lambda)$ under iteration of $f_{\lambda}$.  
\end{lemma}

Even though this lemma is well-known in complex dynamics, for the convenience of the reader we include a sketch of the proof.  Those who wish to see an alternative proof can consider \cite[Lemma 3.1(2)]{GAUTHIER2}.

\begin{proof}
It suffices to prove the statement when the (repelling or attracting) periodic cycle is a fixed point.  Indeed, Lemma \ref{prop 3.1.1} allows for a simple adaptation of both of the proofs below the case of a periodic point of higher period.

It follows from the definition of normal families that we can assume without loss of generality that $n_0 = 0$.  By the Implicit Function Theorem and conjugation of $f_\lambda$ by a suitable holomorphically varying M\"obius transformation we can also assume without loss of generality that $z = 0$ is a repelling fixed point of $f_\lambda(z)$ for all $\lambda$ in some small open neighborhood $U \subset \Lambda$ of $\lambda_0$.   If we let $r > 0$ be sufficiently small then for each $\lambda \in U$ the topological annulus 
\begin{align*}
    A_\lambda = f_\lambda(\mathbb{D}_r) \setminus \mathbb{D}_r
\end{align*}
will serve as a fundamental domain for the dynamics of $f_\lambda$ near $0$.   In other words, for any $\lambda \in U$ and any $z$ with $0 < |z| < r$ there is a unique natural number $k(\lambda,z)$ such that $f_\lambda^{k(\lambda,z)}(z) \in A_\lambda$.  Because $f_\lambda$ is Lipshitz continuous for any $\lambda \in U$ we have that $k(\lambda,z) \rightarrow \infty$ as $z \rightarrow 0$.

Our assumption that $f_\lambda^{n_0}$ maps $a(\lambda_0)$ non-persistently onto the repelling fixed
point and our assumption that $n_0 = 0$ implies that $a(\lambda_0) = 0$ and $a(\lambda) \neq 0$ for all $\lambda \neq \lambda_0$ chosen sufficiently close to $\lambda_0$.   In particular for any $\lambda$ sufficiently close to $\lambda_0$ there exists a natural number $\ell(\lambda) = k(\lambda,a(\lambda))$ such that $f_\lambda^{\ell(\lambda)}(a(\lambda)) \in A_\lambda$. For each $\lambda \in U$ and any $n \in \mathbb{N}$ we have that $f^{-n}_\lambda(A_\lambda)$ is also a topological annulus surrounding $0$.   Therefore, it follows from the intermediate value theorem applied to $|a(\lambda)|$ (or other basic topological considerations) that for any $n \in \mathbb{N}$ there exists $\lambda_n \in U$ such that $\ell(\lambda_n) = n$.

We now suppose for contradiction that $\lambda_0$ is a passive parameter.   Restricting $U$ to a smaller neighborhood of $\lambda_0$, if necessary, we can suppose that $\lambda \mapsto f_\lambda^n(a(\lambda))$ is a normal family on $U$.   In particular this
implies that there is some subsequence $(n_j)_j$ of the sequence of all natural numbers such that $f_\lambda^{n_j}(a(\lambda))$ converges uniformly on compact subsets of $U$ to a holomorphic function $g(\lambda)$.   Note that for $\lambda = \lambda_0$ we have that $f_{\lambda_0}^{n_j}(a(\lambda_0)) = f_{\lambda_0}^{n_j}(0) = 0$ for all $j$ implying that $g(0) = 0$.   On the other hand, as explained in the previous paragraph, for each $j$ there exists $\lambda_j \in U$ such that $\ell(\lambda_j) = n_j$.  Moreover
we must have $\lambda_j \rightarrow 0$ as $j \rightarrow \infty$ because for all $\lambda$ outside of a given neighborhood of $\lambda_0$ we have that $|a(\lambda)|$ is uniformly bounded below, implying that  $\ell(\lambda)$ is uniformly bounded above for such parameters.  Because the convergence of
$f_\lambda^{n_j}(a(\lambda))$ to $g(\lambda)$ is uniform on compact subsets of $U$ and because $f_{\lambda_k}^{n_k}(a(\lambda_k)) \in A_\lambda$ for each index $k$ we therefore have
\begin{align*}
|g(0)| = \lim_{k \rightarrow \infty} |f_{\lambda_k}^{n_k}(a(\lambda_k))| \geq r > 0.
\end{align*}
This contradicts that $g(0) = 0$.

Now suppose that $a(\lambda_0)$ is in the basin of attraction of an attracting fixed point for $f_{\lambda_0}(z)$.   As in the previous case, we can assume without loss of generality that $z=0$ is an attracting fixed point for $f_\lambda(z)$ for all $\lambda$ in some sufficiently small neighborhood $U$ of $\lambda_0$.   Restricting $U$ further if necessary, there is a uniform choice of $r > 0$ such that for
all $\lambda \in U$ the disc $\mathbb{D}_r(0)$ is in the basin of attraction of $0$ under iteration of $f_\lambda(z)$.   By hypothesis
there is a natural number $n_0$ such that $|f_{\lambda_0}^{n_0}(a(\lambda_0))| < r/2$.  By continuity we will then have
that $|f_{\lambda_0}^{n_0}(a(\lambda))| < 3r/4$ for all $\lambda$ in some open neighborhood $V \subset U$ of $\lambda_0$.  It is then
immediate that the sequence of functions $\lambda \mapsto f_{\lambda}^n(a(\lambda))$ converges uniformly to $g(\lambda) \equiv 0$ on
all of $V$.  Therefore $\lambda_0$ is a passive parameter.
\end{proof}

\begin{lemma}\label{LEM:2PERIODS_IMPLIES_ACTIVE} Let $k$ and $\ell$ be distinct natural numbers.   Suppose that in any open neighborhood of parameter $\lambda_0$ there
exist parameters $\lambda_1$ and $\lambda_2$ such that 
\begin{itemize}
    \item[(i)] $a(\lambda_1)$ is in the basin of attraction for an attracting periodic
point of prime period $k$ for $f_{\lambda_1}(z)$ and
\item[(ii)] $a(\lambda_2)$ is in the basin of attraction for an attracting periodic
point of prime period $\ell$ for $f_{\lambda_2}(z)$.
\end{itemize}
Then $\lambda_0$ is an active parameter for the marked point $a(\lambda)$  under iteration of $f_\lambda$.
\end{lemma}

\begin{proof}
Without loss of generality and by Lemma \ref{prop 3.1.1},  we can assume $k=1$ and $\ell>1$. 
For a contradiction, let's assume that $\lambda_0$ is a passive parameter for the marked point $a(\lambda)$ under iteration of $f_\lambda(z)$. Then, there is a connected open neighborhood $U$ of $\lambda_0$ on which $( f_\lambda^n(a(\lambda)))_{n=1}^\infty$ forms a normal family. Then, $(f_\lambda^{\ell \cdot n}(a(\lambda)))_{n=1}^\infty$ has a sub-sequence $(f_\lambda^{\ell \cdot n_k}(a(\lambda)))_{k=1}^\infty$ that converges locally uniformly to a holomorphic function $g(\lambda)$ on $U$. Denote by $z_\bullet$ the attracting fixed point for $f_{\lambda_1}$ that has $a(\lambda_1)$ in its basin of attraction.   By the
Implicit Function Theorem, there is an open neighborhood $V \subset U$ of $\lambda_1$
and a holomorphic function $z_{\rm attr}: V \rightarrow \hat{\mathbb{C}}$ such
that $z_{\rm attr}(\lambda_1) = z_\bullet$ and $z_{\rm attr}(\lambda)$ is an attracting fixed point of $f_\lambda$ for all $\lambda \in V$.   It follows for some potentially smaller open neighborhood $W \subset V$ of $\lambda_1$ that we have $\displaystyle g(\lambda) = z_{\rm attr}(\lambda)$ for all $\lambda \in W$.
Meanwhile, $g(\lambda_2)$ is one of the points in the attracting cycle of period $\ell > 1$ for $f_{\lambda_2}(z)$.

Now consider the family $(f_\lambda^{\ell \cdot n+1}(a(\lambda)) )_{n=1}^\infty$. The sub-sequence  $(f_\lambda^{\ell \cdot n_k+1}(a(\lambda)))_{k=1}^\infty$ converges locally uniformly to the holomorphic function $f_\lambda \left(g(\lambda)\right)$ on $U$.
Notice that 
\begin{align*}
    f(g(\lambda)) = g(\lambda) = z_{\rm attr}(\lambda)
\end{align*}
on $W$ and hence $f(g(\lambda)) = g(\lambda)$ on all of $U$ by the Identity Theorem.   On the other hand
$f(g(\lambda_2)) \neq g(\lambda_2)$ because $g(\lambda_2)$ is a periodic point of period greater than $1$.
This is a contradiction.
\end{proof}

\begin{definition}
    We say that a point $b\in \mathbb{C}_{\infty}$ is \index{exceptional}{\em exceptional} for a rational function ${f: \mathbb{C}_{\infty}\rightarrow \mathbb{C}_{\infty}}$ if the cardinality of the set $f^{-m}(\{b\})$ is less than or equal to two for all $m\in \mathbb{N}$. A marked point $b(\lambda)$ is called \index{persistently exceptional} {\em persistently exceptional} for a holomorphic family of rational maps $\{f_\lambda\}_{\lambda \in \Lambda}$ if $b(\lambda)$ is exceptional for $f_\lambda$ for all $\lambda \in \Lambda$.
\end{definition}

The following lemma plays a key role in our proofs of Theorems A and B.

 \begin{lemma}[Lyubich, \cite{Lyubich1983}]\label{Lemma 1.9.1}
     Let $\{f_\lambda\}_{\lambda \in \Lambda}$ be a family of rational maps which depends holomorphically on the parameter $\lambda$. Let $a(\lambda)$ and $b(\lambda)$ be marked points for $f_\lambda$ with $b(\lambda)$ not being persistently exceptional for $f_\lambda$. If $\lambda_0$ is an active parameter for $a(\lambda)$ under $f_\lambda$, then
     \[\lambda_0\in \overline{\{\lambda\in \Lambda: f^m_\lambda(a(\lambda))=b(\lambda)\ \text{for some}\ m\in \mathbb{N}\}}.\]
 \end{lemma}
\noindent
 In addition to the original source \cite{Lyubich1983}, we also refer the reader to \cite[Lemma V.2]{ChangRoederShrock} for a proof of Lemma \ref{Lemma 1.9.1}


\subsection{Applications of the Active-Passive dichotomy to the renormalization map $R_{z,t,q}(w)$.}\label{SUBSEC:APPLICATIONS_TO_RG_MAP}

We will now interpret our renormalization mapping $R_{z,t,q}(w)$ and Theorem D in the context of Section \ref{section 1.9}.     Note that the degree of $R_{z,t,q}(w)$ in $w$ drops if either $z=0$ or $t \in \{1,1/(1-q)\}$.   Throughout this paper we will only work with $R_{z,t,q}(w)$ for fixed choices of $t \in (0,1) \cup (1,\infty)$ and $q \in \mathbb{N}_{\geq 2}$.   For such fixed $t,q$ we will let $z \in \mathbb{C} \setminus \{0\}$ serve the role of the varying parameter.   These restrictions on the parameters allow us to avoid any possible drop in degree and therefore we obtain a holomorphic family of rational maps in the sense of Section \ref{section 1.9} with $z$ serving the role of the parameter $\lambda$ varying over the parameter space $\Lambda = \mathbb{C} \setminus \{0\}$.

Consider two marked points $a(z) := z$ and $b(z) := 1/(1-q)$.   Theorem D gives that the Lee-Yang zeros
for the level $n$ rooted and unrooted Cayley tree correspond to 
\begin{align*}
\{z \in \mathbb{C} \setminus \{0\} \, : \, R_{z,t,q}^n(a(z)) = b(z)\} \quad \mbox{and} \quad \left\{z \in \mathbb{C} \setminus \{0\} \, : \, \left(\hat{R}_{z,t,q} \circ R_{z,t,q}^{(n-1)}\right)(a(z)) = b(z)\right\},
\end{align*}
respectively.
Note that we have used that $z=0$ is never a Lee-Yang zero for when $t > 0$ which is why we are able to assume $z \in \mathbb{C} \setminus \{0\}$ in the above formula.
Therefore, one has the following expressions for the sets $\mathcal{B}(t,q)$ and $\hat{\mathcal{B}}(t,q)$ that were defined in Equation (\ref{DEF_SETS_B}) above:
\begin{align*}
\mathcal{B}(t,q) &= \overline{\{z \in \mathbb{C} \setminus \{0\} \, : \, R_{z,t,q}^n(a(z)) = b(z) \, \mbox{for some $n \in \mathbb{N}$}\}} \quad \mbox{and} \\
\hat{\mathcal{B}}(t,q) &= \overline{\left\{z \in \mathbb{C} \setminus \{0\} \, : \, \left(\hat{R}_{z,t,q} \circ R_{z,t,q}^{(n-1)}\right)(a(z)) = b(z) \, \mbox{for some $n \in \mathbb{N}$}\right\}},
\end{align*}
respectively.

For fixed $t \in (0,1) \cup (1,\infty)$ and $q \in \mathbb{N}_{\geq 2}$ we denote the active locus of the marked point $a(z) = z$ under the holomorphic family $R_{z,t,q}(w)$ by $\mathcal{A}(t,q)$.  

\begin{lemma} \label{LEM:RELATING_ACTIVITY_LOCUS_TO_ACCUMULATION_LOCUS}
    For any $t \in (0,1) \cup (1,\infty)$ and $q \in \mathbb{N}_{\geq 2}$ we have $\mathcal{A}(t,q) \subset \mathcal{B}(t,q)$ and $\mathcal{A}(t,q) \subset \hat{\mathcal{B}}(t,q)$.
\end{lemma}

\begin{proof}
We will first prove the containment $\mathcal{A}(t,q) \subset \mathcal{B}(t,q)$.
By Lemma \ref{Lemma 1.9.1} it suffices to show that $b(z) = 1/(1-q)$ is not persistently exceptional for $R_{z,t,q}(w)$.  We will show that for $z=1$ the point $b(1)$ is not exceptional
for $R_{1,t,q}(w)$.   The critical values of $R_{1,t,q}(w)$ are $0$ and $\infty$ and therefore
the critical values of the second iterate $R^{2}_{1,t,q}(w)$ consist of 
\begin{align*}
0, \quad \infty, \quad R_{1,t,q}(0) = t^2, \quad \mbox{and} \quad \frac{((q-2) t+1)^2}{(q-1)^2 t^2},
\end{align*}
each of which is non-negative (or infinite).
Therefore $b(1) = 1/(1-q) < 0$ is not a critical value of $R^{2}_{1,t,q}(w)$ implying that it has
four preimages under $R^{2}_{1,t,q}(w)$.   We conclude that $b(1)$ is not exceptional
for $R_{1,t,q}(w)$ and therefore that $b(z)$ is not persistently exceptional for $R_{z,t,q}(w)$.

The containment $\mathcal{A}(t,q) \subset \hat{\mathcal{B}}(t,q)$ will also follow from Lemma \ref{Lemma 1.9.1}.   Note that for any $z \neq 0$ the critical values of $\hat{R}_{z,t,q}(w)$ are $0$ and infinity.   Since $\hat{R}_{z,t,q}(w)$ is a rational map of degree three
and $b(q) = 1/(1-q)$ is not a critical value, one of the three preimages of $b(q)$ will not be exceptional for $R_{z,t,q}(w)$.   Denoting
that preimage by $c(q)$ we apply Lemma \ref{Lemma 1.9.1} to see that 
\begin{align*}
\mathcal{A}(t,q) \subset \overline{\{z \in \mathbb{C} \setminus \{0\} \, : \, R_{z,t,q}^{(n-1)}(a(z)) = c(z) \, \mbox{for some $n \in \mathbb{N}$}\}} \subset \hat{\mathcal{B}}(t,q).
\end{align*}
\end{proof}

\begin{lemma} \label{LEM:RELATING_ACTIVITY_LOCUS_TO_ACCUMULATION_LOCUS2}
For any $t \in (0,1) \cup (1,\infty)$ and $q \in \mathbb{N}_{\geq 2}$ we have 
\begin{align*}
   \mathcal{B}(t,q) \cap (0,\infty) = \hat{\mathcal{B}}(t,q) \cap (0,\infty) = \mathcal{A}(t,q) \cap (0,\infty).  
\end{align*}
\end{lemma}

\begin{proof}
Lemma \ref{LEM:RELATING_ACTIVITY_LOCUS_TO_ACCUMULATION_LOCUS} gives $\mathcal{A}(t,q) \cap (0,\infty) \subset \mathcal{B}(t,q) \cap (0,\infty) $ and $\mathcal{A}(t,q) \cap (0,\infty) \subset \hat{\mathcal{B}}(t,q) \cap (0,\infty)$.   Therefore, it suffices to show that
if $z \in (0,\infty) \setminus \mathcal{A}(t,q)$ then $z \not \in  \mathcal{B}(t,q)$ and $z \not \in  \hat{\mathcal{B}}(t,q)$.

We first show that if $z \in (0,\infty) \setminus \mathcal{A}(t,q)$ then $z \not \in  \mathcal{B}(t,q)$.
Let $z_{\bullet}\in (0,\infty)$ be a passive parameter for the marked point $a(z)=z$ under $R_{z,t,q}(w)$. Then, by definition, there is an open neighborhood $U$ of $z_{\bullet}$ such that the sequence $g_n:U\rightarrow \mathbb{C}_{\infty}$ of functions defined by 
     \[{g_n(z):= R^n_{z,t,q}(a(z))=R^n_{z,t,q}(z)}\] forms a normal family. For the sake of contradiction, let's assume that  $z_{\bullet}\in \mathcal{B}(t,q)$. Then we can find a sub-sequence $(n_k)_{k\in \mathbb{N}}$ and points $z_{n_k}\in U$ such that:
     \begin{enumerate}
         \item [(i)] $(g_{n_k})$ converges uniformly on compact subsets of $U$,
         \item[(ii)] $g_{n_k}(z_{n_k})=R^{n_k}_{z_{n_k},t,q}(z_{n_k})= 1/(1-q)$ for all $k\in \mathbb{N}$, and
         \item[(iii)] $\lim_{k\rightarrow \infty}z_{n_k}=z_{\bullet}$. 
     \end{enumerate}
     Let $g(z)$ be the uniform limit of $g_{n_k}(z)$. Notice $t > 0, q\geq 2$ and $z_\bullet >0$ imply that
     $g_{n_k}(z_\bullet) > 0$ for each $k$ and hence that the limit satisfies $g(z_\bullet) \geq 0$.   On the other hand, because $g_{n_k}$ converges uniformly to $g$, and $z_{n_k}\rightarrow z_\bullet$, we have ${g_{n_k}(z_{n_k})\rightarrow g(z_{\bullet})}$.  This is a contradiction because $g(z_\bullet) \geq 0$ and $g_{n_k}(z_{n_k}) = 1/(1-q) <0$ for each $k$. 

     The proof that if $z \in (0,\infty) \setminus \mathcal{A}(t,q)$ then $z \not \in  \hat{\mathcal{B}}(t,q)$ is quite similar, 
     except that one replaces the sequences of functions $g_n(z)= R^n_{z,t,q}(a(z))$ with a new sequence $\hat{g}_n(z) = \left(\hat{R}_{z,t,q} \circ R^{(n-1)}_{z,t,q}\right)(z)$.   We leave the details to the reader.
\end{proof}

\subsection{Plotting Numerical Approximations to Lee-Yang Zeros Using The {Active-Passive} Dichotomy.}\label{section 1.10}

Lemma \ref{LEM:RELATING_ACTIVITY_LOCUS_TO_ACCUMULATION_LOCUS} gives that $\mathcal{A}(t,q) \subset \mathcal{B}(t,q)$ and $\mathcal{A}(t,q) \subset \hat{\mathcal{B}}(t,q)$.   Although we are primarily interested in the sets $\mathcal{B}(t,q)$ and $\hat{\mathcal{B}}(t,q)$, the active locus $\mathcal{A}(t,q)$ for the marked point $a(z) = z$ under the holomorphic family $R_{z,t,q}(w)$ can be studied dynamically.   Moreover, we can also use the computer to make plots of $\mathcal{A}(t,q)$ that are far more detailed than the plots of Lee-Yang zeros given in Figure \ref{Fig 1.3}.  While the plots of $\mathcal{A}(t,q)$ may be missing some points of $\mathcal{B}(t,q)$ or $\hat{\mathcal{B}}(t,q)$ they still may be informative about those two sets of Lee-Yang zeros.

Arnaud Ch\'eritat helped us by using his computer software to plot the active and passive parameters of the marked point $a(z)=z$ for $R_{z,t,q}$, for fixed $q$ and $t$, in the complex $z$-plane. Figures \ref{AC1} and \ref{AC2} show the active loci related to $t=0.26$ and $t=0.5$ for the $3$-state Potts model. The points in black are supposed to be active parameters and the points in white are supposed to be passive parameters. The reader should compare Figures \ref{AC1} and \ref{AC2} with Figure \ref{Fig 1.3}.

\begin{remark}\label{REM:CONJ_DESCR}
We have made computer images of the active locus for $q = 3$ at many values of $t \in (0,1)$.  They all seem to have the following
property:
\begin{itemize}
\item[(i)] If $0 < t \leq t_2(3)$ we have that $(\mathbb{C}\setminus \{0\})
\setminus \mathcal{A}(t,3)$ consists of two connected components, one
of which contains $z=0$ in its closure and the other which contains $z = \infty$
in its closure.   See Figure \ref{AC1} for an example.
\item[(ii)] If $t > t_2(3)$ then $(\mathbb{C}\setminus \{0\}) 
\setminus \mathcal{A}(t,3)$ consists of a single connected component that contains both $z=0$ and $z= \infty$ in its closure.
See Figure \ref{AC2} for an example.
\end{itemize}
Since the active locus $\mathcal{A}(t,q) \subset \mathcal{B}(t,q)$ is expected to provide a ``good approximation'' to $\mathcal{B}(t,q)$, one may also expect that Properties {\rm (i)} and {\rm (ii)} should also hold
for $\mathcal{B}(t,q)$.  For the $q \geq 3$ state Potts model on the classical $\mathbb{Z}^d$ lattices, this behaviour for the accumulation locus of Lee-Yang zeros
is expected to be true.   It is a generalization of  
the conjectural discription for the accumulation locus of Lee-Yang zeros for the Ising Model, which is believed to be the whole unit circle for $0 < t < t_c$ and a single arc on the circle containing $z=-1$ for $t > t_c$.   See, e.g. \cite[Section 1.3]{Limiting}.

For the Ising Model on the 
the Cayley Tree, the conjectural description for the
accumulation locus of Lee-Yang zeros was proved in \cite[Theorem A]{Limiting}.
We intend to further study whether {\rm (i)} and {\rm (ii)} hold for the
accumulation locus of Lee-Yang zeros associated to the $q \geq 3$ state Potts model on the Cayley
Tree in a future work.  
\end{remark}

It is difficult to numerically compute the Lee-Yang zeros of the partition function when $t>1$. Therefore, we did not include any images analogous to Figure \ref{Fig 1.3} when $t>1$ (the antiferromagnetic case). Instead, we include the plots of the active locus for the marked point $a(z)=z$ for $R_{z,t,3}(w)$ at $t=6$ in Figure \ref{t=6a} and Figure \ref{t=6b}. The Lee-Yang zeros for $\Gamma_n$ accumulate to this active locus as $n\rightarrow \infty$.

We briefly describe here the method used by Ch\'eritat to numerically compute the active locus for a marked point $a(\lambda)$ under a holomorphic family of mappings $f_\lambda(z)$.   A rectangular region in the complex $\lambda$ plane is divided into pixels.  Each pixel
is interpreted as a small square in the complex plane whose side length is $r > 0$.   One also picks a parameter $0 < \theta < 1$, a small
parameter $\epsilon > 0$, and a threshold $m_0 \in \mathbb{N}$.  These are ``tuned'' by the user to get a reasonably looking plot.

To determine if the pixel associated to $\lambda = \lambda_0$ should be colored black (active), white (passive), or red (undecided) the program checks each of the following three conditions for $m=1,2,3,\ldots$.  Once one of the conditions below is met for the pixel 
associated to $\lambda_0$ the program moves on to check the next pixel.

\begin{enumerate}
\item  The program computes $F_m(\lambda) := f_\lambda^m(a(\lambda))$ together with a numerical approximation to the derivative $\left\vert \frac{d(F_m(\lambda) - \lambda)}{d\lambda}\right\vert_{\lambda = \lambda_0}$.   If
\begin{align*}
    |F_m(\lambda_0) - a(\lambda_0)| < r \cdot \theta \left| \frac{d(F_m(\lambda) - \lambda)}{d\lambda} \vert_{\lambda = \lambda_0}\right|
\end{align*}
then it is assumed that $a(\lambda_0)$ is close to being periodic and, when that happens it is likely that $a(\lambda_0)$ is periodic repelling.  In this case the program declares $\lambda_0$ to be active and the pixel is colored black.

\item If Step 1 failed at iterate $m$ then the program checks if 
\begin{align*}
\left|  \frac{d f^m_\lambda(z)}{d z} \vert_{z = a(\lambda_0)}\right| < \epsilon.
\end{align*}
If this holds then the program decides that $a(\lambda_0)$ is probably in the basin of attraction of an attracting periodic cycle.   In this case, the program declares $\lambda_0$ to be passive and the pixel is colored white.

\item If Steps 1 and 2 both fail and if $m$ has reached the threshold $m_0$ then the program stops and colors the pixel red, indicating that it could not decide if $a(\lambda)$ should be interpreted as being active or passive at the
parameter $\lambda_0$.
\end{enumerate}

\begin{figure}[h]
    \centering
    \includegraphics[width=0.35\linewidth]{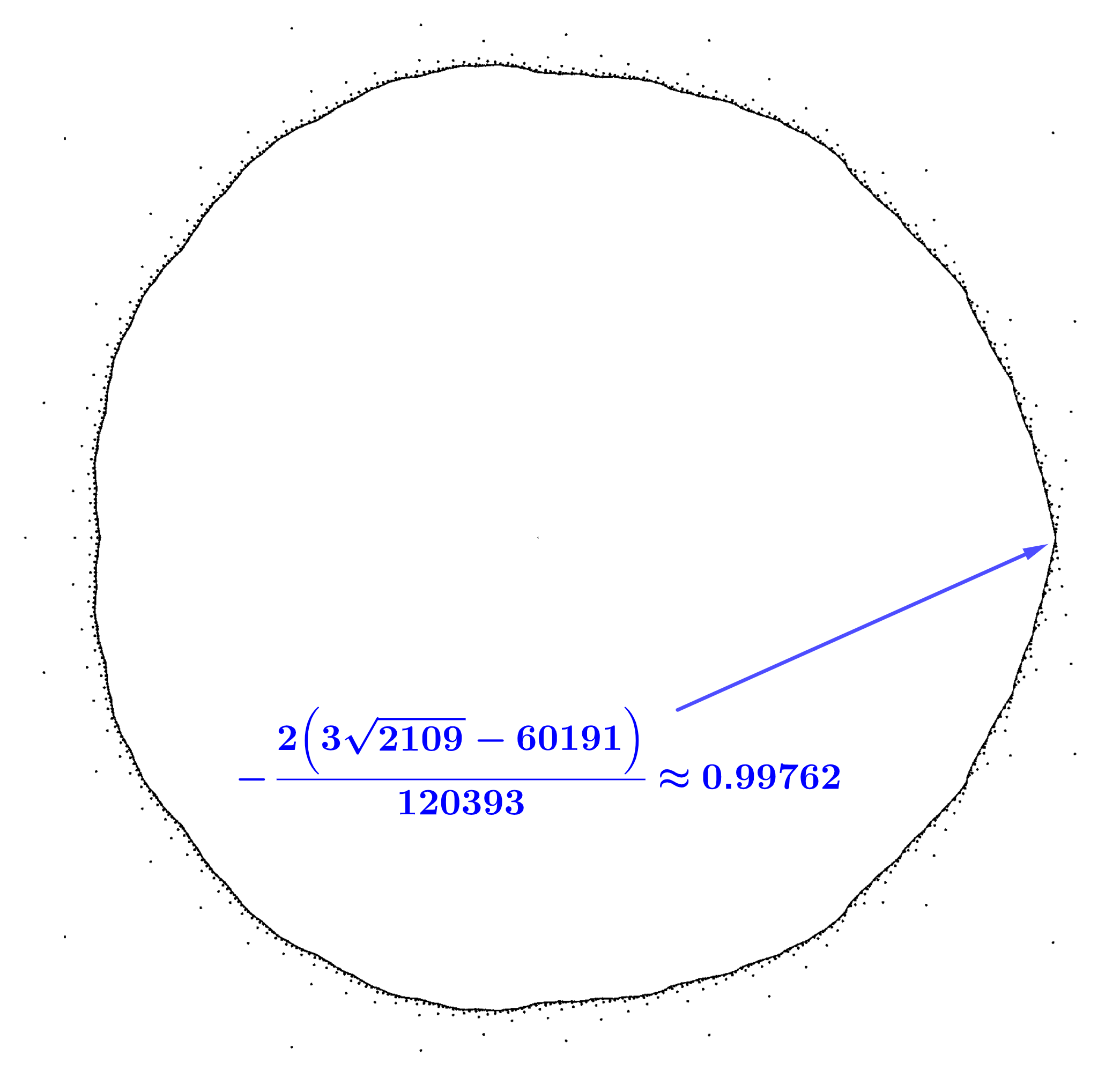}
    \caption{The active locus for the marked point $a(z)=z$ for $R_{z,t,3}(w)$ at $t=0.26$. Points in the active locus are black and the points in the passive locus are white. As explained in Section \ref{section 1.10} the Lee-Yang zeros at $t=0.26$ accumulate to the active locus (black) as $n\rightarrow \infty$.}
    \label{AC1}
\end{figure}

\begin{figure}[h]
    \centering
    \includegraphics[width=0.355\linewidth]{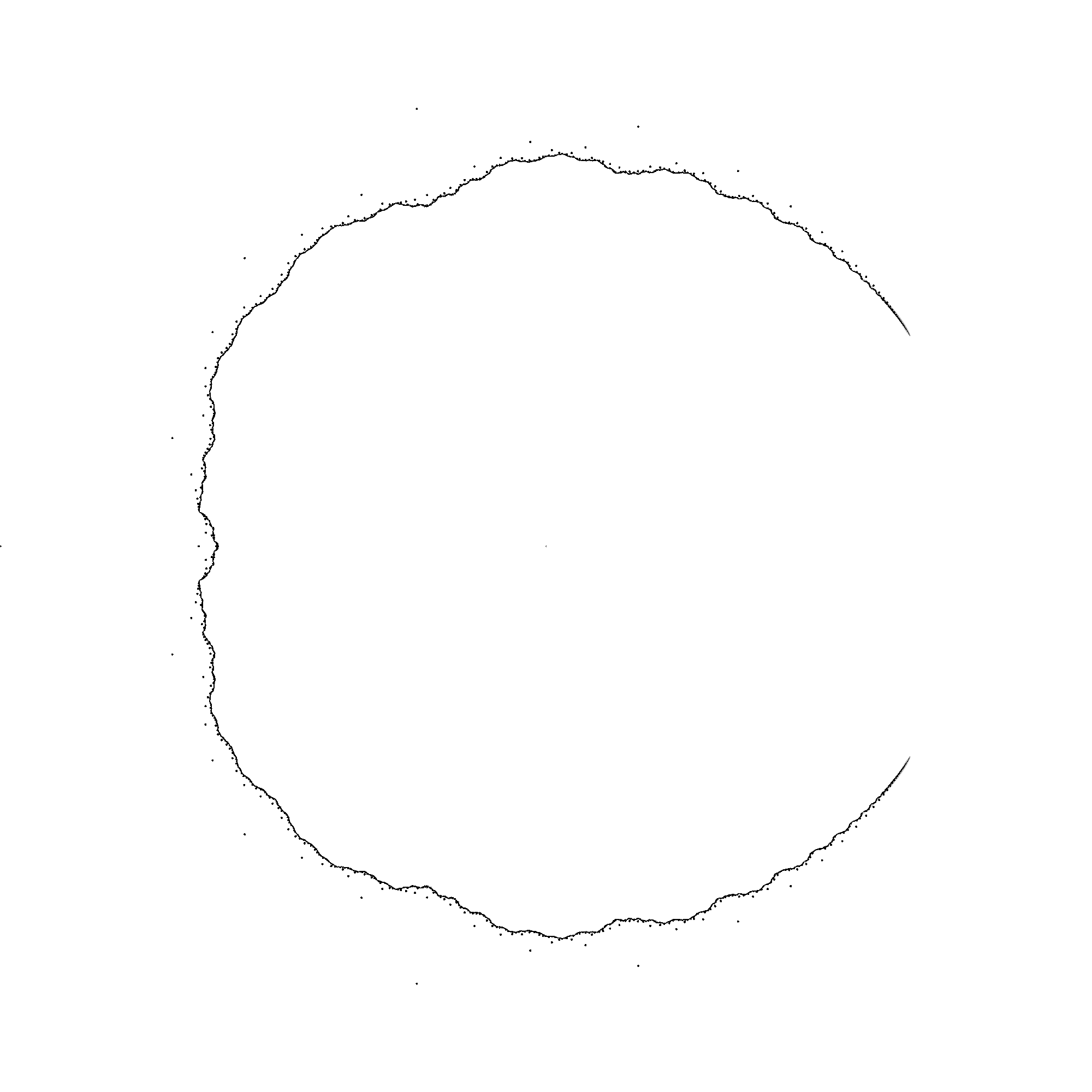}
    \caption{The active locus for the marked point $a(z)=z$ for $R_{z,t,3}(w)$ at $t=0.5$. Points in the active locus are black and the points in the passive locus are white. As explained in Section \ref{section 1.10} the Lee-Yang zeros at $t=0.25$ accumulate to the active locus (black) as $n\rightarrow \infty$.}
    \label{AC2}
\end{figure}

\begin{figure}[h]
    \centering
    \includegraphics[width=0.45\linewidth]{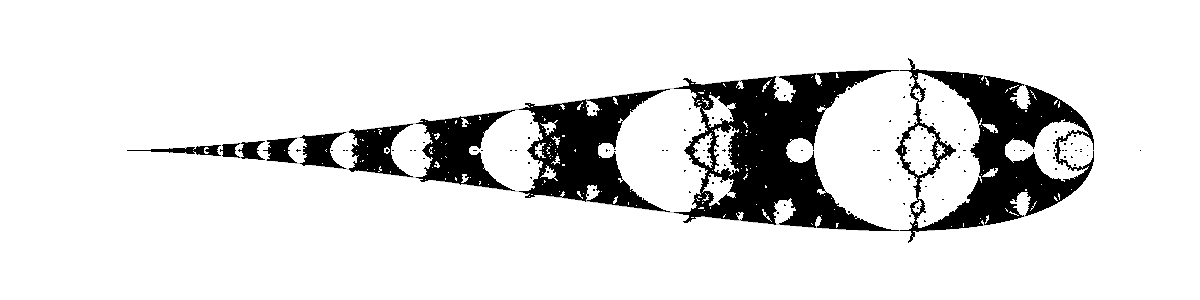}
    \caption{The active locus for the marked point $a(z)=z$ for $R_{z,t,3}(w)$ at $t=6$. Points in the active locus are black and the points in the passive locus are white. The following Figure \ref{t=6b} is a zoomed-in version of this image around $z=1$.}
    \label{t=6a}
\end{figure}

\FloatBarrier

\begin{figure}[h]
    \centering
    \includegraphics[width=0.35\linewidth]{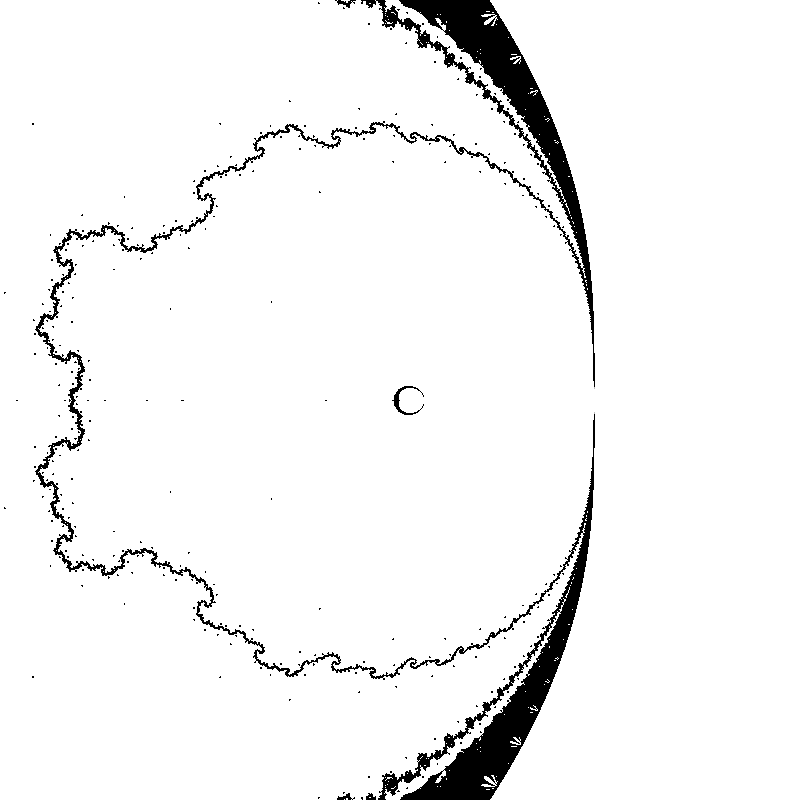}
    \caption{The active locus for the marked point $a(z)=z$ for $R_{z,t,3}(w)$ at $t=6$. Points in the active locus are black and the points in the passive locus are white. This is a zoomed-in version of above Figure \ref{t=6a} around $z=1$.}
    \label{t=6b}
\end{figure}

 \section{PROOF OF THEOREM A}\label{SECTION_PROOF_OF_THMA}

 Note that the cases when $t=0$ and $t=1$ are easily handled using the explicit formula (\ref{EQN:PARTITION_FXN_WHEN_T_EQUALS_01}).   We therefore restrict our attention to $0 < t < 1$ throughout the remainder
 of this section.   We refer the reader to Equation (\ref{EQN:DEF_T1_T2}) for the definitions of $t_1(q)$ and $t_2(q)$.    By Lemma \ref{LEM:RELATING_ACTIVITY_LOCUS_TO_ACCUMULATION_LOCUS2} the proof of Theorem A reduces to the proof of the following theorem.

\begin{theorem}\label{theorem 2}
Suppose $0 <  t < 1$ and let $\mathcal{A}(t,q)$ be the active locus for the marked point $a(z) = z$ under $R_{z,t,q}(w)$. Then 
    \[\mathcal{A}(t,q)\cap (0,\infty) =\left\{
	\begin{array}{ll}
		  \{1\}\ & \mbox{if }\ 0 < t \leq t_1(q),\\
  \{\mathcal{Z}_q(t)\}\ & \mbox{if }\ t_1(q) < t \leq t_2(q),\\
		\ \emptyset\ & \mbox{if }\ t_2(q) < t < 1.
	\end{array}
\right.\]
      Here $\mathcal{Z}_q(t)$ is given in Equation (\ref{EQN:CALZ_Q}) in the statement of Theorem $A$.
   \end{theorem}

We begin with several lemmas.

\begin{lemma}\label{lemma before lemma 4.0.3}
    Let $0<t<1$, $z>0$, and $q\geq 2$. Then the map $\displaystyle R_{z,t,q}(w)$ has the following properties:
    \begin{enumerate}
        \item [(i)] $R_{z,t,q}(w)$ is increasing on $[0,\infty)$,
        \item[(ii)] $R_{z,t,q}(w)$ has at least one fixed point on $[0,\infty)$, and 
        \item[(iii)] $R_{z,t,q}\left([0,\infty)] \right)\subseteq [c,d]$ with $c,d\in \mathbb{R}$ and $0<c<d$. 
    \end{enumerate}
\end{lemma}

\begin{proof}
Simple calculations show that
        \[\frac{\partial R_{z,t,q}}{\partial w}(w) = \frac{2 z(1-t) ((q-1) t+1) ((q-2) t w+t+w)}{((q-1) t w+1)^3}.\] Because $t\in (0,1)$, $z>0$, and $q\geq 2$, we get $\displaystyle \frac{\partial R_{z,t,q}}{\partial w}(w)>0$ for all $w>0$. Thus, $R_{z,t,q}(w)$ is increasing on $[0,\infty)$.  This proves Claim (i).

        Notice that $R_{z,t,q}(0)=zt^2>0$. Furthermore, the horizontal asymptote of $R_{z,t,q}(w)$, \[\displaystyle \lim_{w\rightarrow \infty}R_{z,t,q}(w)=z\left[\frac{1+(q-2)t}{(q-1)t}\right]^2\] is positive and finite. Thus, $R_{z,t,q}(w)$ must intersect the diagonal at least once. Hence, $R_{z,t,q}(w)$ has at least one positive fixed point, thus proving Claim (ii).

Thus, if we let \[c:= R_{z,t,q}(0)=zt^2\ \ \ \text{and}\ \ \ d:=\lim_{w\rightarrow \infty}R_{z,t,q}(w)=z\left[\frac{1+(q-2)t}{(q-1)t}\right]^2,\] then Claim (iiii) follows from Claims (i) and (ii). 
\end{proof}

 Let $c,d\in \mathbb{R}$ with $c<d$. Let $f:[c,d]\rightarrow [c,d]$ be a differentiable function such that $f'(x)>0$ for all $x \in [a,b]$, and with $f(a) > a$, and $f(b) < b$.    A fixed point $x_{\rm fix}$ for $f$ is
 {\em repelling} if $f'(x_{\rm fix}) > 1$, {\em neutral} if $f'(x_{\rm fix}) = 1$, and {\em attracting} if $f'(x_{\rm fix}) < 1$.  The following Lemma is well-known, so we omit the proof.

\begin{lemma}\label{Lemma 4.0.3}
    Let $f:[c,d]\rightarrow [c,d]$ be as described above.  Let $A_f$, $N_f$ and $R_f$ be the set of all attracting, neutral, and repelling fixed points of $f$ respectively. Then for any $x_0\in (c,d)\backslash (R_f\cup N_f\cup A_f)$, the orbit of $x_0$ under $f$ converges to a point in $A_f\cup N_f$. Furthermore, if there is a point $x_N\in N_f$ and an open interval $I\ni x_0, x_N$ with no fixed point of $f$ in $I\backslash \{x_N\}$ such that 
    \begin{enumerate}
        \item [(i)] $x_0<x_N\ and\ f\vert_I(x)>x$ 
        or
        \item[(ii)] $x_0>x_N\ and\ f\vert_I(x)<x$,
    \end{enumerate}
    then the orbit $\{f^n(x_0)\}_{n=0}^\infty$ of $x_0$ converges to $x_N$.
\end{lemma}

\begin{lemma}\label{lemma 4.0.5}
    Let $t\in (0,1)$ be fixed. Suppose $z>0$ is an active parameter for the marked point $a(z)=z$ under $R_{z,t,q}(w)$. Then one of the following two cases occurs.
    \begin{enumerate}
        \item[(i)] The marked point $a(z)=z$ is a repelling fixed point of $R_{z,t,q}(w)$.
        \item[(ii)]  $R_{z,t,q}(w)$ has a positive neutral fixed point.
    \end{enumerate}
\end{lemma}

\begin{proof}
    Let $t\in (0,1)$ be fixed. Then for any fixed $z>0$, by Lemma \ref{lemma before lemma 4.0.3}, the function $R_{z,t,q}(w)$ is increasing on $[0,\infty)$ with at least one fixed point. Moreover, $R_{z,t,q}(w)$ maps the compact interval $\displaystyle \left[c,d\right]$ to itself. Therefore, from Lemma \ref{Lemma 4.0.3} we see that $a(z)=z$ is a repelling fixed point of $R_{z,t,q}(w)$, the orbit of $z$ under $R_{z,t,q}$ converges to a neutral fixed point, or the orbit of $z$ under $R_{z,t,q}$ converges to an attracting fixed point of $R_{z,t,q}(w)$.   By Lemma \ref{Lemma 3.1.1}, the third option corresponds to passive behavior of the marked point $a(z)$, so that (i) or (ii) must occur.
\end{proof}

\begin{lemma}\label{lemma 4.0.6}
    Let $t \in (0,\infty)\setminus \{1\}$ and $z > 0$. Then the marked point $a(z)=z$ is a fixed point of $R_{z,t,q}(w)$ if and only if $z=1$.
\end{lemma}

\begin{proof}
A direct calculation gives $R_{1,t,q}(1) = 1$.   Conversely, $t \geq 0$ and $z \geq 0$ imply that $t+z+(q-2)tz \geq 0$, and $1+(q-1)tz>0$ which allows us to simplify
\[R_{z,t,q}(z)=z\left[\frac{t+z+(q-2)tz}{1+(q-1)tz}\right]^2=z \] to \[t+z+(q-2)tz=1+(q-1)tz\] by taking a square root.  This implies $z(1-t)=1-t$. Because $t\neq1$, we get $z=1$.
\end{proof}

\begin{lemma} \label{lemma 4.0.7}
    Let $z > 0$, $t\in (0,1)$ and $q \geq 2$. If $\displaystyle R_{z,t,q}(w)$ has a neutral fixed point at some $w=w_N(z)>0$, then $t \in (0,t_2(q)]$
    and 
    $z=N_{\pm}(t,q)$, where 
\begin{align*}
N_{\pm}(t,q):= 
\frac{\left(\substack{
(-27 (q-1)^2 t^4+18 \left(q^2-3 q+2\right) t^3+\left(q^2+14 q-14\right) t^2 +2 (q-2) t+1))
 \\
 \pm \sqrt{(t-1) ((q-1) t+1) \left(9 (q-1) t^2-(q-2) t-1\right)^3}
}\right)}
     {8 t ((q-2) t+1)^3}.
 \end{align*}

     Moreover, for each value of $t \in (0,t_2(q)]$ both values $N_\pm(t,q)$ are positive.
\end{lemma}

\noindent \textit{Remark:} $\mathcal{Z}_q(t)$ from Theorem A is equal to $N_-(t,q)$.   Figure \ref{fig 4.1} shows a plot $N_\pm(t,q)$ in the case that $q = 3$.

\begin{figure}[h]
    \centering
    \includegraphics[width=0.6\linewidth]{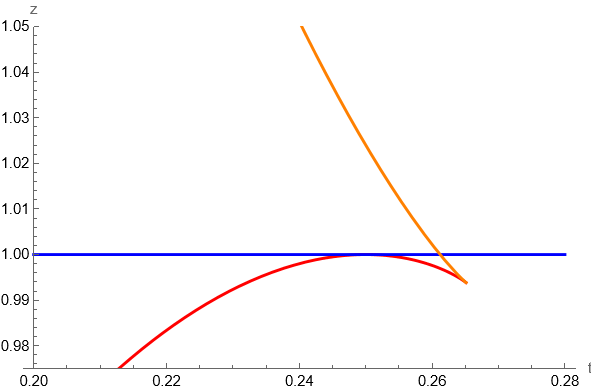}
    \caption{Plots of $N_+(t,3)$ (orange), $N_-(t,3)$ (red), and $1$ (blue) versus $t$. Compare with Figure \ref{graph of $y=Z(t)$}.}
    \label{fig 4.1}
\end{figure}

\begin{proof}
Throughout the proof, fix $q \geq 2$.
    By definition, $w_N>0$ is a neutral fixed point of $R_{z,t,q}(w)$ if and only if 
\begin{align*}
    R_{z,t,q}(w_N)=w_N \qquad \mbox{and} \qquad \frac{d}{dw}R_{z,t,q}(w)\mid_{w=w_N(z)}=\pm 1.
\end{align*}
However, by Lemma \ref{lemma before lemma 4.0.3} we find that $\displaystyle \frac{d}{dw}R_{z,t,q}(w)=1$.
Therefore the above conditions reduce to the condition that $R_{z,t,q}(w) - w$ has a multiple root at $w_N$.   Clearing denominators, this is equivalent to the following polynomial in $w$ having a multiple root at some $w_N$:
\begin{align*}
P_{z,t,q}(w):=z((q - 2)tw + t + w)^2 - w((q - 1)tw + 1)^2.
\end{align*}
Taking the discriminant of $P_{z,t,q}(w)$ with respect to the variable $w$ we find that $R_{z,t,q}$ has a neutral fixed point
if and only if
\begin{align*}
    -(-1 + t)^2 (1 - t + q t)^2 z \, Q(z,t,q) = 0,
    \end{align*}
where the last factor $Q(z,t,q)$ is the following quadratic polynomial in $z$:
\begin{align*}
  Q(z,t,q) := &  \left(4 q^3 t^4-24 q^2 t^4+48 q t^4-24 (q-2) t^3+12 (q-2) q t^3+12 q t^2-32 t^4-24 t^2+4 t\right)z^2 + \\
  & \left(-q^2 t^4+3 q^2 t^2+2 q t^4-6 (q-2) t^3+6 (q-2) q t^3-6 q t^2+6 q t-t^4+6 t^2-12 t+3\right) z + \\
  &4 q t-4 t.
\end{align*}
Since $t \in (0,1)$, $z > 0$, and $q \geq 2$ the only way that
the discriminant of $P_{z,t,q}(w)$ vanishes is when $Q(z,t,q)$ vanishes.
The reader can check that $Q(z,t,q)$ vanishes precisely under the conditions stated in the conclusion of this lemma.

The square root in the formula for 
 $N_\pm(t)$ is non-negative if and only if $t \in (0,t_2(q)]$.   Since the denominator of $N_\pm(t,q)$ does not vanish for $t > 0$ we have that $N_\pm(t,q)$ both vary continuously as $t$ varies over the interval $(0,t_2(q)]$.  When $t = 1/(1+q)$ one finds that 
 \begin{align*}
 N_-(t,q) = 1 \qquad \mbox{and} \qquad N_+(t,q) = \frac{(q-1)(q+1)^3}{(2q-1)^3},
 \end{align*}
 both of which are positive.   Meanwhile,
 substituting $z=0$ into the polynomial $P_{z,t,q}(w)$ yields $-w (1 + (-1 + q) t w)^2$ which has no solutions when $w,t > 0$.  Therefore $N_\pm(t,q)$ remain positive for all $t \in (0,t_2(q)]$.
\end{proof}

We will need the following lemma to complete the proof of Theorem \ref{theorem 2}.   Since the proof involves some technical calculations, we will give it in Appendix \ref{APPENDIX_TECHNICAL_LEMMAS}.

\begin{lemma}\label{LEM:TECHNICAL} 
We have:
\begin{enumerate}
    \item [(i)] Suppose $z=N_+(t,q)$. Then, for any $t \in (0,t_2(q))$, the orbit of the marked point $a(z) = z$ under iteration of $R_{z,t,q}$ converges to an attracting fixed point $w_A(z)$ of $R_{z,t,q}(w)$.
     \item [(ii)] Suppose $z=N_-(t,q)$.  Then:
     \begin{enumerate}
         \item [(a)] If $t \in (0,t_1(q))$, then the orbit of the marked point $a(z) = z$ under iteration of $R_{z,t,q}$ converges to an attracting fixed point $w_A(z)$ of $R_{z,t,q}(w)$.
         \item[(b)] If $t \in (t_1(q),t_2(q))$, then the orbit of the marked point $a(z) = z$ under iteration of $R_{z,t,q}$ converges to a neutral fixed point  $w_N(N_-(t,q))$ of $R_{z,t,q}(w)$.
     \end{enumerate}
     \item[(iii)] If $t \in (t_1(q),t_2(q))$ and $z < N_-(t,q)$ is sufficiently close to $N_-(t,q)$, then the orbit of $z$ under iteration of $R_{z,t,q}$ converges to an attracting fixed point $w_A(z)$ of $R_{z,t,q}(w)$.
\end{enumerate}
    \end{lemma}

\begin{proof}[Proof of Theorem \ref{theorem 2}]
Suppose $0<t< 1$. By Lemma \ref{lemma 4.0.5}, if $z>0$ is an active parameter for the marked point $a(z)=z$ under $R_{z,t,q}(w)$, then either
  \begin{enumerate}
      \item [(i).] $a(z)=z$ is a repelling fixed point for $R_{z,t,q}(w)$, and in this case, that repelling fixed point must be $w=1$ by Lemma \ref{lemma 4.0.6}, or 
\item[(ii).] $R_{z,t,q}(w)$ has a neutral fixed point $w_N(z)>0$. Then by Lemma \ref{lemma 4.0.7}, $t \in (0,t_2(q)]$ and $z=N_{\pm}(t,q)$.
 \end{enumerate}

\noindent In Case $(i)$, a simple calculation gives
 \[R'_{1,t,q}(1)=\frac{2-2 t}{(q-1) t+1}.\] By solving the inequality $ R'_{1,t,q}(1)>1$ for $t$ we get $\displaystyle 0\leq t< t_1(q)$. Moreover, when $\displaystyle 0\leq t< t_1(q)$, $R_{z,t,q}(w)$ maps $a(z)=z$ non-persistently onto this repelling fixed point $w=1$ as $a(z)=z$ is not fixed for $z\neq 1$. Then by Lemma \ref{Lemma 3.1.1}, $z=1$ is an active parameter for the marked point $a(z)=z$ under the map $R_{z,t,q}(w)$. 
 If $\displaystyle t=t_1(q)$, then $z=1$ continues to be an active parameter for the marked point $a(z)=z$ as the active locus is closed. If $\displaystyle t>t_1(q)$, then $w=1$ is an attracting fixed point. So, in that case, $a(z)=z$ is passive at $z=1$.

We now consider Case $(ii)$.  Note that when $t = t_2(q)$ we have $N_+(t,q) = N_-(t,q)$.  To finish the proof of Theorem \ref{theorem 2} we must check:
\begin{enumerate}
    \item [(I).] For any $0<t< t_2(q)$ we have that $z=N_+(t,q)$ is passive for the marked point $a(z)=z$ under $R_{z,t,q}(w)$.
    \item[(II).] For any $0<t< t_1(q)$ we have that $z=N_-(t,q)$ is passive for the marked point $a(z)=z$ under $R_{z,t,q}(w)$, and for any $t_1(q) \leq t\leq  t_2(q)$ we have that $z=N_-(t,q)$ is active for the marked point $a(z)=z$ under $R_{z,t,q}(w)$.
\end{enumerate}

Suppose $\displaystyle 0<t< t_2(q)$. If $z=N_+(t,q)$, then Lemma \ref{LEM:TECHNICAL}(i) gives that the orbit of $z$ converges to an attracting fixed point of $R_{z,t,q}(w)$. Thus, by Lemma \ref{Lemma 3.1.1}, $z=N_+(t,q)$ is a passive parameter for the marked point $a(z)=z$ under $R_{z,t,q}(w)$.   Similarly, when $0<t< t_1(q)$ and $z=N_-(t,q)$, Lemma \ref{LEM:TECHNICAL}(ii(a)) gives that the orbit of $z$ converges to an attracting fixed point of $R_{z,t,q}(w)$, making it a passive parameter for the marked point $a(z)=z$ under $R_{z,t,q}(w)$ by Lemma \ref{Lemma 3.1.1}.

 Suppose $t_1(q) <t < t_2(q)$ and $z=N_-(t,q)$.
 By Lemma \ref{LEM:TECHNICAL}(ii(b)), the orbit of $z_0=N_-(t,q)$ converges to the neutral fixed point $w_N(N_-(t,q))$. Furthermore, there exists an attracting fixed point $0 < w_A<w_N(N_-(t,q))$. This attracting fixed point varies continuously with respect to $z$ around $N_-(t,q)$. We will denote this dependence on $z$ by $w_A \equiv w_A(z)$.  By Lemma \ref{LEM:TECHNICAL}(iii) if $z< z_0$ is sufficiently close to $z_0$, then the orbit of $z$ converges to $w_A(z)$. 
Thus, for any subsequence $\left(R_{z,t,q}^{n_k}(z)\right)_{k=1}^\infty$ of $\left(R_{z,t,q}^{n}(z)\right)_{n=1}^\infty$ and any neighborhood $U$ of $z_0=N_-(t,q)$, we can find a further subsequence $\left(R_{z,t,q}^{n_{k_j}}(z)\right)_{j=1}^\infty$ of $\left(R_{z,t,q}^{n_k}(z)\right)_{k=1}^\infty$ and a point $z_1\in U$ such that for all $j \in \mathbb{N}$
\[\left|R_{z_1,t,q}^{n_{k_j}}(z_1)-R_{z_0,t,q}^{n_{j_k}}(z_0)\right| > \frac{w_N(z_0)-w_A(z_1)}{2}.\] 
Because $w_N(z_0) - w_A(z_0) > 0$ and because $w_A(z)$ depends
continuously on $z$ we can choose $z_1 \in U$ sufficiently close to $z_0$ so that
the right-hand side of the above inequality is positive.

Therefore the family $\left\{R_{z,t,q}^{n}(z)\right\}_{n=1}$ of functions is not normal at $z_0~=~N_-(t,q)$.  Hence, the marked point $a(z)=z$ under $R_{z,t,q}(w)$ is active at $z_0=N_-(t,q)$. Because the active locus is a closed set,  the marked point $a(z)=z$ under $R_{z,t,q}(w)$ is also active at $z_0=N_-(t,q)$ when $\displaystyle t= t_1(q)$ and $\displaystyle t=t_2(q)$.

This completes the proof of Theorem \ref{theorem 2}.    
\end{proof}

 \section{Proof of Theorem B (Antiferromagnetic Case, $q\geq 3$)}\label{SECTION_PROOF_OF_THMB}
Recall that Theorem B concerns the antiferromagnetic regime $J<0$. Since the physical values of temperature are $T>0$, this corresponds to the temperature-like parameter ${\displaystyle t=e^{-\frac{J}{T}}>1}$ and the field-like parameter $\displaystyle z=e^{-\frac{h}{T}}>0$.   By Lemma \ref{LEM:RELATING_ACTIVITY_LOCUS_TO_ACCUMULATION_LOCUS2} the proof of Theorem B reduces to the proof of the following theorem.

\begin{theorem}\label{Theorem 5.0.2}
Suppose $t > 1, q \geq 3$ and let $\mathcal{A}(R_{z,t,q})$ be the active locus for the marked point $a(z) = z$ under $R_{z,t,q}(w)$. Then 
\begin{align}\label{EQN:ACTIVE_LOCUS_THMB}
 \mathcal{A}(R_{z,t,q})\cap (0,\infty) =
 \begin{cases}
		  \emptyset \ & \mbox{if }\ 1 <  t < t_3(q) \\
  \{z_c^{\pm}(t,q)\}\ & \mbox{if }\ t_3(q)\leq t.
	\end{cases}
\end{align}     
Here $t_3(q) > 1$ and $z_c^{\pm}(t,q)$ are given in Equations (\ref{EQN:T3}) and (\ref{EQN:Z_C_PM}) as in the statement of Theorem~B.
%
\end{theorem}

Therefore, the rest of this section is devoted to the proof of Theorem \ref{Theorem 5.0.2}. The reader can see the plots of $z_c^{\pm}(t,3)$ in Figure (\ref{Fig 1.6}). The main challenge is that $R_{z,t,q}(w)$ is not increasing on $[0,\infty)$ when $t>1$.   However, the second iterate $R^2_{z,t,q}(w) = R_{z,t,q}\left(R_{z,t,q}(w)\right)$ of $R_{z,t,q}(w)$ is increasing on $[0,\infty)$ when $t>1$.   
By Lemma \ref{prop 3.1.1} the active locus for the marked point $a(z)=z$ under $R_{z,t,q}(w)$ is the same as for the second iterate $R_{z,t,q}^2(w)$ so we can therefore work extensively with $R_{z,t,q}^2(w)$ throughout the proof of Theorem \ref{Theorem 5.0.2}.

\begin{lemma} \label{lemma 8}
    Let $z>0$, $t>1$,  and $q\geq 2$. The map
    \begin{align*}
        R^2_{z,t,q}(w) :=R_{z,t,q}\left(R_{z,t,q}(w) \right)
        &= \frac{z \left(z ((q-2) t+1) ((q-2) t w+t+w)^2+t ((q-1) t w+1)^2\right)^2}{\left((q-1) t z ((q-2) t w+t+w)^2+((q-1) t w+1)^2\right)^2}
    \end{align*}
    has the following properties.
      \begin{enumerate}
        \item [(i)] $R^2_{z,t,q}(w)$ is increasing on $[0,\infty)$,
        \item[(ii)] $R^2_{z,t,q}(w)$ has at least one fixed point on $[0,\infty)$, and 
        \item[(iii)] $R^2_{z,t,q}\left([0,\infty)] \right)\subseteq [c,d]$ with $c,d\in \mathbb{R}$ and $0<c<d$. 
    \end{enumerate}
\end{lemma}

\begin{proof}
 The interval $[0,\infty)$ is forward invariant under $R_{z,t,q}(w)$ so, by the chain rule it suffices to prove that $\frac{\partial}{\partial w}R_{z,t,q}(w) < 0$  when $z> 0, \ t > 1,$ and $q\geq 2$ and when $w \geq 0$.  Expressing everything in terms of $z, \ s=t-1 > 0, \ p = q-2 > 0$, and $w>0$ we find:
        \begin{align*}
        \frac{\partial}{\partial w}R_{z,t,q}(w) =  \frac{-2 z s (2 + p + s + p s) (1 + s + w + p w + p s w)}{(1 + w + p w + s w + p s w)^3} < 0.
        \end{align*}
        This proves Claim (i). To prove Claim (ii), notice that
\[R^2_{z,t,q}(0) = \frac{z \left(t^2 z ((q-2) t+1)+t\right)^2}{\left((q-1) t^3 z+1\right)^2}>0,\] and 
\begin{align*}
    \lim_{w\rightarrow \infty}R^2_{z,t,q}(w) = 
    \frac{z \left(t^3 \left(q^3 z+q^2 (1-6 z)+2 q (6 z-1)-8 z+1\right)+3 (q-2)^2 t^2 z+3 (q-2) t z+z\right)^2}{(q-1)^2 t^2 \left((q-2)^2 t^2 z+t (2 q z+q-4 z-1)+z\right)^2}
\end{align*}
which is positive and finite. Hence, $R_{z,t,q}^2(w)$ has at least one positive fixed point by the Intermediate Value Theorem.   Since $R_{z,t,q}^2(w)$ is increasing on $[0,\infty)$ we can prove Claim (iii) by letting $c:=R_{z,t,q}^2(0)$ and $d:= \lim_{w\rightarrow \infty}R^2_{z,t,q}(w)$.     
\end{proof}

\begin{lemma}\label{tech lemma 1 for theorem B}
    For fixed $z > 0$, $t>1$ and $q\geq 2$, if the map $\displaystyle R^2_{z,t,q}(w)$ has a neutral fixed point, then 
    \begin{align*}
        t\geq t_{3}(q) > 1 \qquad \mbox{and} \qquad z=z_c^{\pm}(t,q),
        \end{align*}
        where the formulae for $t_{3}(q)$ and $z=z_c^{\pm}(t,q)$  are given in Equations (\ref{EQN:T3}) and (\ref{EQN:Z_C_PM}).  
        
        Furthermore, for all $t \geq t_3(q)$ we have 
        \begin{itemize}
            \item[(a)] $0< z_c^-(t,q)\leq z_c^+(t,q)$ with the equality $z_c^-(t,q) = z_c^+(t,q)$ only when $t=t_{3}(q)$.
            \item[(b)] If $q \geq 3$ then $z_c^+(t,q) < 1$.
        \end{itemize}
\end{lemma}

\begin{proof}
The proof is quite similar to the proof of Lemma \ref{lemma 4.0.7} given in the previous section.   We express the condition that
$R^2_{z,t,q}(w)$ has a neutral fixed point as the existence of a multiple root
of a polynomial
$P_{z,t,q}(w)$ which is obtained from the equation $R^2_{z,t,q}(w) = w$ by clearing denominators.
The discriminant of $P_{z,t,q}(w)$ with respect to $w$ equals
\begin{align}\label{EQN:DISCRIM}
    z^8 (t-1)^{16} ((q-1)t+1)^{16}\, Q_1(z,t,q) Q_2(z,t,q)^3,
    \end{align}
where 
{\footnotesize
\begin{align*}
    Q_1(z,t,q) = z [t (27 (q-1)^2 t^3-18 (q-2) (q-1) t^2-q (q+14) t-2 q+14 t+4)-1]
    +4 t z^2 ((q-2) t+1)^3+4 (q-1) t, 
\end{align*}
}
and
{\footnotesize
\begin{align}\label{EQN:Q2}
    Q_2(z,t,q) = z (-(q-1)^2 t^4+6 (q-2) (q-1) t^3+3 ((q-2) q+2) t^2+6 (q-2) t+3)
    +4 t z^2 ((q-2) t+1)^3+4 (q-1) t.
\end{align}
}
Therefore, we find that $R_{z,t,q}$ has a neutral fixed point
if and only if one of the factors in Equation~(\ref{EQN:DISCRIM}) vanishes. 
 The first three factors clearly don't vanish for $z > 0, t > 1$, and $q \geq 2$.   We claim that $Q_1(z,t,q)$ also does not vanish for this range
 of the parameters $z,t,q$.  Expressing everything in terms of $z, s=t-1 > 0, p = q-2 \geq 0$, and $w>0$ we find:
 \begin{align}\label{EQN:BIG_Q1_EXPRESSION}
 Q_1 &= 4 p^3 s^4 z^2+16 p^3 s^3 z^2+24 p^3 s^2 z^2+16 p^3 s z^2+4 p^3
   z^2+27 p^2 s^4 z+12 p^2 s^3 z^2+90 p^2 s^3 z  \nonumber \\
   & +36 p^2 s^2 
   z^2 
   +107 p^2 s^2 z+36 p^2 s z^2+52 p^2 s z+12 p^2 z^2+8 p^2
   z+54 p s^4 z+198 p s^3 z  \nonumber \\
   &+12 p s^2 z^2+252 p s^2 z+24 p s
   z^2 
   +124 p s z+4 p s+12 p z^2+16 p z+4 p+27 s^4 z \\
   &+108 s^3
   z+144 s^2 z+4 s z^2+72 s z+4 s+4 z^2+8 z+4 > 0. \nonumber
 \end{align}
 Note that $Q_2(z,t,q)$ is quadratic in $z$ with the coefficients depending on $t$ and $q$.   Solving $Q_2(z,t,q)~=~0$ for $z$ in terms of $t$ and $q$ we find precisely the values $z_c^\pm(t,q)$ given in Formula (\ref{EQN:Z_C_PM}).   Note that the term under the square root in formula for $z_c^\pm(t,q)$ vanishes at $t = t_3(q)$ and that it is negative for $1 < t < t_3(q)$ and strictly positive for $t > t_3(q)$.   In particular, when $t > t_3(q)$ we have $z_c^-(t,q) < z_c^+(t,q)$.

We now prove the additional claim (a) by checking that for all $t \geq t_3(q)$ we have $0 < z_c^-(t,q)$.
 When $t=t_3(q)$ we have
 \begin{align}\label{EQN:Z_C_PM_TEQUALS_T3}
 z_c^+(t,q) = z_c^-(t,q) = \frac{-27 q^3+153 q^2-297 q+198+\left(9 q^2-35 q+35\right) \sqrt{9 q^2-32 q+32}}{2 (q-1)},
 \end{align}
 which one can check is positive for all $q \geq 2$.   It is clear from the formulae for $z_c^\pm(t,q)$ that for fixed $q \geq 2$ they are continuous functions of $t \geq t_3(q)$.   Suppose for contradiction that there were some
 $t_4 > t_3(q)$ with $z_c^-(t_4,q)~<~0$.  In this case the Intermediate Value Theorem would give that
 there exists $t_5 \in [t_3(q),t_4]$ with $z_c^-(t_5,q) = 0$.   We obtain a contradiction by noting that we would then have
 $Q_2\left(z_c^-(t_5,q),t_5,q\right) = Q_2(0,t_5,q) = 0$ while a direct calculation gives that
 $Q_2(0,t,q) = 4(q-1)t > 0$ for all $q \geq 2$ and $t > 0$.

 To prove additional claim (b) it suffices to show that $Q_1(z,t,q) > 0$ for all $z \geq 1$, $t \geq 1$, and $q \geq 3$.  If one computes $Q_1(1+r,1+s,3+p)$ one finds a polynomial whose constant term is $72$ and all of whose monomials appear with a ``plus" sign.   (We have omitted the many line expression which is similar to Equation (\ref{EQN:BIG_Q1_EXPRESSION}) above, but the reader can check it in a computer algebra software package.)   
\end{proof}

In the following statement, we extend the range of allowable values of $q$ from $\mathbb{N}_{\geq 3}$ to the interval $[3,\infty)$ so that
methods involving continuity can be used in the proof.   

\begin{lemma}\label{tech lemma 2 for theorem B}
    Define sets $\mathcal{R}_1$ and $\mathcal{R}_2$ as follows:
    \begin{align*}
        \mathcal{R}_1&:= \{(z,t,q)\in \mathbb{R}^3: t > 1, \ q\geq 3, \ \mbox{and if $t \geq t_3(q)$ then} \ 0 < z<z_c^-(t,q)\ \text{or}\ z>z_c^+(t,q)\}\\        
        \mathcal{R}_2&:= \{(z,t,q)\in \mathbb{R}^3: t \geq t_3(q), \ q\geq 3, \ z_c^-(t,q)<z<z_c^+(t,q)\}.
    \end{align*}
Then,
\begin{enumerate}
    \item[(i)] For all $(z,t,q)\in \mathcal{R}_1$ the second iterate of the renormalization map, $R^2_{z,t,q}(w)$, has only one real fixed point. This fixed point is in $(0,\infty)$ and it is an attracting fixed point of $R_{z,t,q}(w)$.

    \item[(ii)] For all $(z,t,q)\in \mathcal{R}_2$ the second iterate of the renormalization map, $R^2_{z,t,q}(w)$, has three real fixed points each of which is in $(0,\infty)$. One of these fixed points is a repelling fixed point of $R_{z,t,q}(w)$. The other two fixed points correspond to an attracting period 2-cycle of $R_{z,t,q}(w)$.
\end{enumerate}
\end{lemma}

\begin{proof}
We first prove that $\mathcal{R}_1$ and $\mathcal{R}_2$ are each path connected.
We begin with $\mathcal{R}_1$.   One can easily check from the formula (\ref{EQN:T3}) that for $q \geq 3$ one has $t_3(q) \geq 3$.  Therefore, we have the containment
\begin{align*}
B:=\{(z,t,q) \in \mathbb{R}^3 \, : \, z > 0, 1 < t < 3, \mbox{ and } q \geq 3\} \subset \mathcal{R}_1.
\end{align*}
Note that $B$ is convex and hence path connected.   Moreover, for each
fixed choice of real $q = q_0 \geq 3$ the slice
\begin{align*}
\mathcal{R}_1 \cap \{q = q_0\} &:= \{(z,t,q_0)\in \mathbb{R}^3: t > 1 \ \mbox{and if $t \geq t_3(q_0)$ then} \ 0 < z<z_c^-(t,q_0)\ \text{or}\ z>z_c^+(t,q_0)\}
\end{align*}
has non-trivial intersection with $B$.   Therefore it suffices to show
 for each real $q_0 \geq 3$ that the slice $\mathcal{R}_1 \cap \{q = q_0\}$ is path connected.
 We will describe an explicit path within $\mathcal{R}_1 \cap \{q = q_0\}$ from any point $(z_0,t_0,q_0) \in \mathcal{R}_1 \cap \{q = q_0\}$ to the point $(1,2,q_0) \in B$.   There are three possibilities:

\vspace{0.1in}
\noindent
{\bf Option 1:} $1 < t_0 < t_3(q_0)$.   The straight line path between $(z_0,t_0,q_0)$ and $(1,2,q_0)$ remains in $\mathcal{R}_1 \cap \{q = q_0\}$.

\vspace{0.1in}
\noindent
{\bf Option 2:} $t_0 \geq t_3(q_0)$ and $z_0 > z_c^+(t_0,q_0)$.   Since $q_0 \geq 3$, Lemma \ref{tech lemma 1 for theorem B} gives that $z_c^+(t_0,q_0) < 1$.   We can therefore form a path from $(z_0,t_0,q_0)$ to $(1,2,q_0)$ by first connecting $(z_0,t_0,q_0)$ to $(1,t_0,q_0)$ using a straight-line path (varying only $z$) and then connecting from $(1,t_0,q_0)$ to $(1,2,q_0)$ by a straight line path (varying only $t$).  Each of these line segments is in $\mathcal{R}_1 \cap \{q = q_0\}$.

\vspace{0.1in}
\noindent
{\bf Option 3:} $t_0 \geq t_3(q_0)$ and $0 < z_0 < z_c^-(t_0,q_0)$.   By Lemma \ref{tech lemma 1 for theorem B} we have that $z_c^-(t,q_0) > 0$ for all $t \geq t_3(q_0)$.  Moreover, it is clear from the expression for $z_c^-(t,q_0)$ that it is a continuous function of $t \geq t_3(q)$.   Therefore, there exists a positive $M > 0$ such that for all $t \in [t_3(q_0),t_0]$ we have $z_c^-(t,q_0) \geq M$.   We can therefore form a path within $\mathcal{R}_1 \cap \{q = q_0\}$ from $(z_0,t_0,q_0)$ to $(1,2,q_0)$ by a concatenation of three straight line paths.   One first
connects from $(z_0,t_0,q_0)$ to $(M/2,t_0,q_0)$ by varying only $z$.  One then connects from $(M/2,t_0,q_0)$ to $(M/2,2,q_0)$ by varying only $t$.   One finally connects from $(M/2,2,q_0)$ to $(1,2,q_0)$ by varying only $z$.

\vspace{0.1in}
We have therefore proved that $\mathcal{R}_1$ is path connected.   We will now prove that $\mathcal{R}_2$ is path-connected.   For each $t \geq t_3(q)$ let 
\begin{align*}
    z_c^{\rm mid}(t,q) = \frac{1}{2}(z_c^-(t,q) + z_c^+(t,q)).
\end{align*}
For each $t > t_3(q)$ one has $z_c^{\rm mid}(t,q) \in \mathcal{R}_2$.   In particular, given any $3 \leq q_0 < q_1$
we can form a path in $\mathcal{R}_2$ connecting the slice $\mathcal{R}_2 \cap \{q = q_0\}$ to the slice
$\mathcal{R}_2 \cap \{q = q_1\}$ using the mapping $q \mapsto (z_c^{\rm mid}(t_3(q)+1,q),t_3(q)+1,q)$ where $q$ varies over the interval $[q_0,q_1]$.
Therefore, it suffices to prove for each $q = q_0 \geq 3$ that $\mathcal{R}_2 \cap \{q = q_0\}$ is path-connected.   Given any $(z_0,t_0,q_0) \in \mathcal{R}_2 \cap \{q = q_0\}$ we can connect it using a path
within $\mathcal{R}_2 \cap \{q = q_0\}$ to $(z_c^{\rm mid}(t_3(q_0)+1),t_3(q_0)+1,q_0)$ as follows.  One first varies just the $z$-coordinate to connect from $(z_0,t_0,q_0)$ to $(z_c^{\rm mid}(t_0,q_0),t_0,q_0)$.   One then connects from $(z_c^{\rm mid}(t_0,q_0),t_0,q_0)$ to $(z_c^{\rm mid}(t_3(q_0)+1),t_3(q_0)+1,q_0)$ by varying $t$ between $t_0$ and $t_3(q_0)+1$ and letting $z = z_c^m(t,q_0)$ as one does so.

\vspace{0.1in}
We now prove Statements (i) and (ii).   The regions $\mathcal{R}_1$ and $\mathcal{R}_2$ were chosen so that the second iterate of the renormalization mapping $R^2_{z,t,q}(w)$ has only attracting or repelling (but not neutral) real fixed points for all parameters $(z,t,q) \in \mathcal{R}_1 \cap \mathcal{R}_2$.  Therefore, as one varies the parameters $(z,t,q)$ within a given region $\mathcal{R}_1$ or $\mathcal{R}_2$ these real fixed points each vary continuously and their nature (attracting or repelling) remains unchanged.   (This also implies that the same holds for fixed points of the first iterate $R_{z,t,q}(w)$.)  Moreover, Lemma \ref{lemma 8}(iii) gives for all $(z,t,q) \in \mathcal{R}_1 \cup \mathcal{R}_2$ that $R_{z,t,q}^2\left((0,\infty)\right)$ is contained in a compact subset of $(0,\infty)$.  Therefore, as $(z,t,q)$ varies over
$\mathcal{R}_1$ or over $\mathcal{R}_2$ any fixed point of $R_{z,t,q}^2(w)$ that lies in $(0,\infty)$ of one choice of parameters $(z,t,q)$ moves
to a fixed point of $R_{z,t,q}^2(w)$ that again lies in $(0,\infty)$ for the other choice of parameters.
Therefore it suffices to check the assertions of each of the claims (ii) and (iii) for a single choice of parameters in each region.

\vspace{0.1in}
\noindent
{\bf Region $\mathcal{R}_1$:}
Note that $(z,t,q) = (1,3,3) \in \mathcal{R}_1$.
One finds that 
\begin{align*}
    R_{1,3,2}(w)  = \frac{(4 w+3)^2}{(6 w+1)^2} \qquad \mbox{and} \qquad
    R_{1,3,2}^2(w)  =\frac{\left(172 w^2+132 w+39\right)^2}{\left(132 w^2+156 w+55\right)^2}.
\end{align*}
By solving $R^2_{1,3,3}(w)=w$ for $w$ one finds that the only real solution is $w=1$
and that it is also a fixed point of $R_{1,3,3}(w)$.   Moreover, $|R_{1,3,2}'(w)| = \frac{4}{7} < 1$, so the only real fixed point of $R^2_{1,3,3}(w)$ is actually an attracting fixed point of the
first iterate $R_{1,3,3}(w)$.

\vspace{0.1in}
\noindent
{\bf Region $\mathcal{R}_2$:}  
Note that $(z,t,q) = (1/5,8,3) \in R_2$ because
\begin{align*}
z_c^-(8,3) = \frac{9229-119 \sqrt{5593}}{46656} \approx 0.0071 \qquad \mbox{and} \qquad z_c^+(8,3) = \frac{9229-119 \sqrt{5593}}{46656} \approx 0.3886.
\end{align*} We have
\begin{align*}
    R_{1/5,8,3}(w) = \frac{(9 w+8)^2}{5 (16 w+1)^2} \ \ \ \text{and} \qquad 
     R_{1/5,8,3}^2(w) = \frac{\left(1567 w^2+368 w+88\right)^2}{5 \left(368 w^2+352 w+147\right)^2}.
\end{align*} 
By solving $ R_{1/5,8,3}^2(w)=w$ for $w$ one finds the following three real roots: 
\begin{align*}
    w_1 &:= \frac{939-17 \sqrt{2165}}{1058} \approx 0.1399, \qquad
    w_2 :=  \frac{939+17 \sqrt{2165}}{1058} \approx 1.6352, \quad \mbox{and} \\
    w_3 &\approx 0.4412, \quad \mbox{which is the unique real root of $p(x) = -64-139x +79x^2+1280 x^3$.}
\end{align*}
The only real root of $R_{1/5,8,3}(w) =w$ is $w=w_3$ and one finds that $ \left|(R_{1/5,8,3})'(w_3)\right| \approx 1.088>1$. Thus, $w_3$ is a repelling fixed point of $R_{1/5,8,3}(w)$.
Furthermore, one can check that $R_{1/5,8,3}(w_1) = w_2$, $R_{1/5,8,3}(w_2) = w_1$, and
\begin{align*}
\left| (R_{1/5,8,3}^2)'(w_1) \right| = \left| (R_{1/5,8,3}^2)'(w_2) \right| \approx 0.7003 <1.
\end{align*}
Thus, $w_1$ and $w_2$ correspond to an attracting period $2$-cycle of $R_{1/5,8,3}(w)$.
\end{proof}

\begin{proof}[Proof of Theorem \ref{Theorem 5.0.2}]
We will first prove that if $(z,t,q) \in \mathcal{R}_1 \cup \mathcal{R}_2$ then the marked point $a(z) = z$ is passive under
iteration of $R_{z,t,q}(w)$.   By Lemma \ref{prop 3.1.1} it suffices to prove that $a(z)$ is passive under the second iterate $R^2_{z,t,q}(w)$.   Moreover, by Lemma \ref{Lemma 3.1.1}, it suffices to prove that for such choices of parameters we have that $a(z)$ is in the basin of attraction of an attracting fixed point for~$R^2_{z,t,q}(w)$.

Suppose $(z,t,q) \in \mathcal{R}_1 \cup \mathcal{R}_2$.
By Lemma~\ref{lemma 8}, we see that $R^2_{z,t,q}(w)$ satisfies the hypotheses of Lemma \ref{Lemma 4.0.3}.  Lemma \ref{tech lemma 2 for theorem B} implies that $R^2_{z,t,q}(w)$ has no neutral fixed points and also that if $R^2_{z,t,q}(w)$ has a repelling fixed point then it is actually a repelling fixed point of the first iterate $R_{z,t,q}(w)$.   Meanwhile, Lemma \ref{lemma 4.0.6} gives that if $a(z)$ is a fixed point for $R_{z,t,q}(w)$ then $z = 1$.   As in the proof of Theorem \ref{theorem 2} we compute that  
\begin{align*}
R'_{1,t,q}(1)=\frac{2-2 t}{(q-1) t+1}.
\end{align*}
One can then check that for $t > 1$ and $q \geq 3$ we have $-1 < R'_{1,t,q}(1) < 0$.  Therefore, $a(z) = z$ cannot be a repelling fixed point of $R_{z,t,q}(w)$.   By Lemma \ref{Lemma 4.0.3} we therefore have that if $(z,t,q) \in \mathcal{R}_1 \cup \mathcal{R}_2$ then the marked point $a(z) = z$  is in the basin of attraction of an attracting fixed point of $R^2_{z,t,q}(w)$, as desired.   

It remains to show that if $t \geq t_3(q)$ and $z=z_c^\pm(t,q)$ then the marked point $a(z) = z$ is active under $R_{z,t,q}(w)$.   Since the active locus is closed, we can suppose that $t > t_3(q)$.     Then, for any $z > z^+_c(t,q)$ we have that $(z,t,q) \in \mathcal{R}_1$ and Lemma \ref{tech lemma 2 for theorem B} and the discussion in the previous paragraph gives that $a(z)$ is in the basin of attraction of an attracting fixed point for $R_{z,t,q}(w)$.  On the other hand for any $z_c^-(t,q) < z < z_c^+(t,q)$ we have that $a(z)$ is in the basin of attracting of a period two attracting cycle of $R_{z,t,q}(w)$.   Lemma \ref{LEM:2PERIODS_IMPLIES_ACTIVE} then implies that the marked point $a(z)$ is active at the parameter $z = z_c^+(t,q)$.

The proof that $a(z)$ is active at $z = z_c^-(t,q)$ is completely analogous.
\end{proof}

\section{Proof of Theorem C (Antiferromagnetic Case, $q=2$)}\label{SEC:PROOF_THMC}

Like the proof of Theorem B, the proof of Theorem C reduces to the following statement.

\begin{theorem}\label{THM:STATEMENT_ABOUT_ACTIVE_LOCUS_THMC}
Suppose $t > 1, q = 2$ and let $\mathcal{A}(R_{z,t,2})$ be the active locus for the marked point $a(z) = z$ under $R_{z,t,2}(w)$. Then 
\begin{align}\label{EQN:ACTIVE_LOCUS_THCB}
 \mathcal{A}(R_{z,t,2})\cap (0,\infty) =
 \begin{cases}
		  \emptyset \ & \mbox{if }\ 1 < t < 3 \\
  \{1\} \cup \{z_c^{\pm}(t,2)\}\ & \mbox{if }\ 3\leq t.
	\end{cases}
\end{align}     
Here $t_3(2) = 3$ and $z_c^{\pm}(t,2)$ are given in Equations (\ref{EQN:T3}) and (\ref{EQN:Z_C_PM_QIS2}) as in the statement of Theorem~C.
\end{theorem}

\begin{proof}[Sketch of proof]
The proof is quite similar to the proof of Theorem \ref{Theorem 5.0.2} with a couple of small variations.  First note that Lemmas \ref{lemma 8} and \ref{tech lemma 1 for theorem B} both hold for $q = 2$.  We did not include the case $q=2$ in the statement of Lemma \ref{tech lemma 2 for theorem B} because $z_c^+(t,2)$ crosses $z = 1$ when $t = 3$.  This would have complicated our proof that $R_1$ and $R_2$ are connected (as $q$ is varied), which played an important role in the proof of Lemma \ref{tech lemma 2 for theorem B}.   However, if one
considers the following sets 
 \begin{align*}
        \mathcal{S}_1 &:= \{(z,t)\in \mathbb{R}^2: t > 1, \ \mbox{and if $t \geq 3$ then} \ 0 < z<z_c^-(t,2)\ \text{or}\ z>z_c^+(t,2)\}\\        
        \mathcal{S}_2 &:= \{(z,t)\in \mathbb{R}^2: t > 3, \ z_c^-(t,q)<z<z_c^+(t,2)\}.
    \end{align*}
then it is straightforward to prove the analogous claims for them as in Lemma \ref{tech lemma 2 for theorem B}.  (As before, one proves that each set is connected and then checks the claims (i) and (ii) for a single choice of parameters in each set.   The proof that
$S_1$ and $S_2$ are connected is rather straightforward, since one can directly check that $z_c^+(t,2)$ is an increasing function and $z_c^-(t,2)$ is a decreasing function of $t \geq 3$.) 

The proof of Theorem \ref{THM:STATEMENT_ABOUT_ACTIVE_LOCUS_THMC} then concludes in almost the same way as the proof of Theorem \ref{Theorem 5.0.2}, except that since $q =2$ one can have that the marked point $a(z) = z$ does non-persistently 
land on a repelling fixed point $w=1$ when $z=1$.   It does this for all $t \geq 3$, leading to the additional point on $(0,\infty)$ at which the Lee-Yang zeros can accumulate that was described in Theorem C.
\end{proof}

\appendix

\section{Proof of Lemma \ref{LEM:TECHNICAL}.}\label{APPENDIX_TECHNICAL_LEMMAS}

We start with two simple lemmas that will be used in the proof.   In each of these
statements we extend the range of allowable values of $q$ from $\mathbb{N}_{\geq 2}$ to the interval $[2,\infty)$ so that
methods involving continuity can be used in the proofs.   

\begin{lemma}\label{LEM:TRIPLE_FP}
Suppose $z > 0$, $t \in [0,1)$, and $q \geq 2$.   If the renormalization mapping $R_{z,t,q}(w)$ has triple fixed
point (i.e.\ $w$ is a triple root of $R_{z,t,q}(w) = w$) then $t = t_2(q) = \frac{q-2 + \sqrt{q^2 + 32q - 32}}{18(q - 1)}$.
\end{lemma}

\begin{proof}
A polynomial $w^3 + aw^2 + bw+c = 0$ has a triple root if and only if $b = \frac{a^2}{3}$ and $c = \frac{a^3}{27}$.
Using that $z > 0$, $t \in [0,1)$, and $q \geq 2$, we can rewrite the condition that $R_{z,t,q}(w) = w$ as a monic cubic polynomial
of the above form.   It has a triple root if and only if 
\begin{align*}
   \frac{[2 t z ((q-2) t+1)-1]}{ -(q-1)^2 t^2} &= \frac{1}{3}\left(\frac{2 t z ((q-2) t+1)-1}{ -(q-1)^2 t^2} \right)^2 \, \text{and} \,
\frac{t^2 z}{ -(q-1)^2 t^2}  = \frac{1}{27}\left(\frac{2 t z ((q-2) t+1)-1}{ -(q-1)^2 t^2} \right)^3.
\end{align*}
Clearing denominators, the above two equations can be simplified to the conditions that $P_1(z,t,q)~=~0$ and $P_2(z,t,q)~=~0$ where $P_1$ and $P_2$ are polynomials in $z, q$ and $t$. One can then use elimination theory (resultants) to
eliminate $z$ from the above two equations.  One finds the following factorized polynomial in $q$ and $t$:
\[(q-1)^4 (t-1)^2 t^4 (q t-t+1)^2 \left(9 q t^2-q t-9 t^2+2 t-1\right) = 0.\]
Since $z > 0$, $t \in [0,1)$, and $q \geq 2$, the only relevant root comes from the last
factor and it is $t = t_2(q)$.
\end{proof}

\begin{lemma}\label{LEM:WHEN_IS_MARKED_PT_FIXED}
    For any $q \geq 2$ and  $t \in (0,t_2(q)]$
 we have:
    \begin{itemize}
        \item[(a)] when $z = N_+(t,q)$, the point $w=z$ is fixed by $R_{z,t,q}(w)$ if and only if $t = \frac{1-2\sqrt{q-1}}{5-4q}$, and
        \item[(b)]when $z = N_-(t,q)$, the point $w=z$ is fixed by $R_{z,t,q}(w)$ if and only if $t = t_1(q) = \frac{1}{q+1}$.
    \end{itemize} 
\end{lemma}

\begin{proof}
Note that the marked points $a(z) = N_\pm(t,q)$ are positive by Lemma
\ref{lemma 4.0.7}.    
Lemma~\ref{lemma 4.0.6} therefore gives that they are fixed points of $R_{z,t,q}(w)$ if and only if they equal $1$.  The result then follows by solving $N_\pm(t,q)=1$.
\end{proof}

\begin{lemma}\label{LEM:NEUTRAL_AND_ATTRACTING_FP}
    For any $q \geq 2$ and  $t \in (0,t_2(q))$ if $z = N_\pm(t,q)$ then $R_{z,t,q}(w)$ has exactly two fixed points and the non-neutral fixed point is attracting.
\end{lemma}

\begin{proof}
For $t \in (0,t_2(q))$ and $z = N_\pm(t,q)$ the mapping $R_{z,t,q}(w)$ has exactly two fixed points, with one corresponding to a double root of $R_{z,t,q}(w) = w$ and the other corresponding to a simple root of the same equation.  (See the proof of Lemma \ref{lemma 4.0.7}.)  With this range of values for the parameters $t,q$ and $z=N_\pm(t,q)$ we have that $R_{z,t,q}'(w) > 0$ for $w > 0$.   Therefore, the double root corresponds to a neutral fixed point $w_N(t,q)$ and for all $w$ in a sufficiently small interval $I$ containing $w_N(t,q)$ we either have $R_{z,t,q}(w) \geq w$ or for all $w \in I$ 
we have $R_{z,t,q}(w) \leq w$.   In other words, the graph of $R_{z,t,q}(w)$ does not cross the diagonal at the neutral fixed point $w_N(t,q)$.

Note that $R_{z,t,q}(0) = zt^2 > 0$, and that
\[ \lim_{w \to \infty} R_{z,t,q}(w) = \frac{z ((q-2)t + 1)^2}{(q - 1)^2 t^2} \in (0,\infty).\]
Therefore, if we denote the fixed point corresponding to the simple root of $R_{z,t,q}(w) = w$ by $w_2(t,q)$ we must have a small interval $J$ containing $w_2(t,q)$ such that $R_{z,t,q}(w) > w$ for all $w \in J$ satisfying $w < w_2(t,q)$
and such that $R_{z,t,q}(w) < w$ for all $w \in J$ satisfying $w > w_2(t,q)$.   This implies that $0 < R'_{z,t,q}(w_2(t,q)) < 1$
and hence that $w_2$ is an attracting fixed point for $R_{z,t,q}(w)$.  
\end{proof}

For the remainder of the proof we will denote the neutral and attracting fixed points whose existence is
given by Lemma \ref{LEM:NEUTRAL_AND_ATTRACTING_FP} by $w_N(t,q)$ and $w_A(t,q)$.

\begin{proof}[Proof of Lemma \ref{LEM:TECHNICAL}]
{\bf Proof of Claim (i):}   We suppose that $t \in (0,t_2(q))$ and that $z=N_+(t,q)$.   Rather than considering $q \geq 2$ as a natural number we consider it as a real number.  By Lemma~\ref{LEM:WHEN_IS_MARKED_PT_FIXED} the marked point $z= N_+(t,q)$ is a fixed point of $R_{z,t,q}(w)$ if and only if $t = \frac{1-2\sqrt{q-1}}{5-4q}=:t_4(q)$.   A direct calculation gives that $t_4(q) = t_2(q)$ if and only if $q=2$ and that for $q > 2$ we have $0 < t_4(q) < t_2(q)$.   Thus, the two curves $t=t_2(q)$ and $t=t_4(q)$ divide the $qt$-plane, when $t > 0$ and $q\geq 2$, into three regions:
\begin{align*}
\mbox{Region I:} \quad  0 < t < t_4(q), \qquad \mbox{Region II:} \quad t_4(q) < t < t_2(q), \qquad \mbox{and}
\quad        \mbox{Region III:} \quad t > t_2(q).
\end{align*}
Refer to Figure \ref{TechLemma} for an illustration. 
 \begin{figure}[h]
     \centering
     \includegraphics[width=0.6\linewidth]{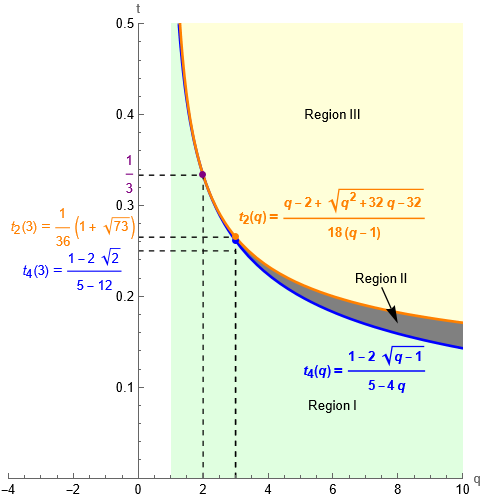}
     \caption{Plot of Regions I, II, and III in proof of Claim (i).}
     \label{TechLemma}
 \end{figure}
Note that Region III is not considered under the hypotheses of Lemma \ref{LEM:TECHNICAL}, so we will ignore it.

Note that every complex fixed point of $R_{z,t,q}(w)$ is real when $z = N_+(t,q)$ and when $(t,q)$ varies over
either Region I or Region II and that one of them is consistently a solution of multiplicity two to the equation $R_{z,t,q}(w) = w$.  Since complex solutions to $R_{z,t,q}(w) = w$ vary continuously with respect to the parameters, this implies that
$w_N(t,q)$ and $w_A(t,q)$ vary continuously as the parameters $(t,q)$ are varied over Region I or over Region II and they are never equal since that would correspond
to a fixed point of multiplicity~$3$ which only happens when $t=t_2(q)$, by Lemma~\ref{LEM:TRIPLE_FP}.  Moreover, by Lemma \ref{lemma before lemma 4.0.3} we have that $R_{z,t,q}\left([0,\infty)\right)$ is compactly contained in $(0,\infty)$ for each choice of parameters, so that any fixed points in $(0,\infty)$ for one choice of parameters cannot leave $(0,\infty)$ for a different choice of parameters in the same region.

Within Regions I and Region II we also have that the marked point $z = N_+(t,q)$ is not fixed by $R_{z,t,q}(w)$ because $t \neq t_4(q)$ in those regions.   Therefore, the order with which $w_N(t,q), w_A(t,q),$ and $N_+(t,q)$ occur on $[0,\infty)$ is constant in each of the two regions.

\vspace{0.1in}
\noindent
{\bf Region I:}
Considering $(t,q) = (0.1,5)$, which lies in Region 1.  For these parameters one can do explicit 
calculations in a computer algebra system (e.g. Mathematica) to find:
\begin{align*}
        z= N_+(0.1,5) &= \frac{141 \sqrt{329}+6457}{4394} \approx 2.051,\\
        w_N(0.1,5) &=  \frac{1}{52} \left(59-3 \sqrt{329}\right) \approx 0.0881717, \quad \mbox{and} \\
        w_A(0.1,5) &= \frac{1}{416} \left(189 \sqrt{329}+3433\right) \approx 16.4931.
\end{align*}
In particular, we find that for all $(t,q)$ in Region I we have the following order
\begin{align*}
w_N(t,q) < z=N_+(t,q) < w_A(t,q).
\end{align*}
In particular, since $R_{z,t,q}(0) > 0$ and since the graph of $R_{z,t,q}(w)$ does not cross the diagonal at $w_N(t,q)$ we find that $R_{z,t,q}(w) > w$ for all $w \in (w_N(t,q),w_A(t,q))$.  This implies that
the orbit of $z = N_+(t,q)$ under iteration of $R_{z,t,q}(w)$ converges to $w_A(t,q)$, as claimed.

\vspace{0.1in}
\noindent
{\bf Region II:}
Considering $(t,q) = (0.15,10)$ which is in Region II.
For these parameters one can again do explicit 
calculations to find:
\begin{align*}
        z= N_+(0.15,10) &= \frac{151 \sqrt{120649}+1880293}{2044416} \approx 0.945376, \\
        w_N(0.15,10) &=  \frac{637-\sqrt{120649}}{2376} \approx 0.121908, \quad \mbox{and} \\
        w_A(0.15,10) &= \frac{799 \sqrt{120649}+327037}{769824} \approx 0.78533.
\end{align*}
In particular, we find that for all $(t,q)$ in Region II we have the following order
\begin{align*}
w_N(t,q) <  w_A(t,q) < z=N_+(t,q).
\end{align*}
Since $\lim_{w \to \infty} R_{z,t,q}(w)$ is finite we have that for all $w > w_A(t,q)$ that
$R_{z,t,q}(w) < w$.  This implies that
the orbit of $z = N_+(t,q)$ under iteration of $R_{z,t,q}(w)$ converges to $w_A(t,q)$, as claimed.

 Finally, suppose that $t = t_4(q)$, corresponding to the boundary between Regions I and II.  It is clear from the above calculations and continuity that in this case $z = N_+(t,q) = w_A(t,q)$.   Thus Claim (i) is proved.

\vspace{0.1in}
{\bf Proof of Claim (ii):}
We suppose that $t \in (0,t_2(q))$ and that $z=N_-(t,q)$.    By Lemma~\ref{LEM:WHEN_IS_MARKED_PT_FIXED} the marked point $z= N_+(t,q)$ is a fixed point of $R_{z,t,q}(w)$ if and only if $t = t_1(q) = 1/(q+1)$.   A direct calculation gives that $t_1(q) = t_2(q)$ if and only if $q=2$ and that for $q > 2$ we have $0 < t_1(q) < t_2(q)$.   Thus, the two curves $t=t_2(q)$ and $t=t_1(q)$ divide the $qt$-plane, when $t > 0$ and $q\geq 2$, into three regions:
\begin{align*}
\mbox{Region I:} \quad  0 < t < t_1(q), \qquad \mbox{Region II:} \quad t_1(q) < t < t_2(q), \qquad \mbox{and}
\quad        \mbox{Region III:} \quad t > t_2(q).
\end{align*}
Refer to Figure \ref{TechLemma2} for an illustration.
\begin{figure}[h]
    \centering
    \includegraphics[width=0.6\linewidth]{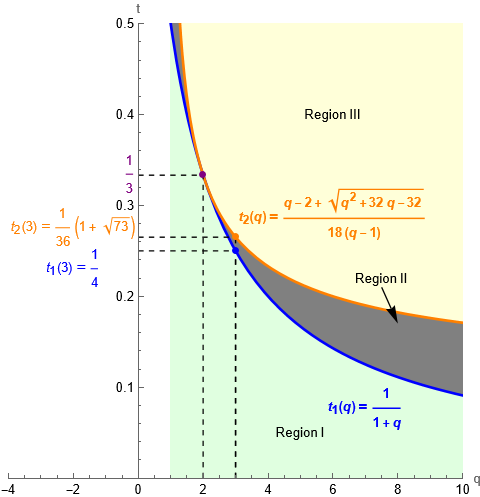}
    \caption{Plot of Regions I, II, and III in proof of Claim (ii).}
    \label{TechLemma2}
\end{figure}
Note that Region III is not considered under the hypotheses of Lemma \ref{LEM:TECHNICAL}, so we will ignore it.

Just as in the proof of Claim (i) the fixed points $w_N(t,q), w_A(t,q),$ and the marked point $z=N_-(t,q)$ vary continuously with $(t,q)$ over Regions I and II, that they never leave $(0,\infty)$, and on a given region they are never equal.    Therefore, the order with which $w_N(t,q), w_A(t,q),$ and $N_-(t,q)$ occur on $[0,\infty)$ is constant in each of the two regions.

\vspace{0.1in}
\noindent
{\bf Region I:}
Consider $(t,q) = (0.2,3)$, which lies in Region I.  For these parameters one can do explicit 
calculations to find:
\begin{align*}
z=N_-(0.2,3) &= \frac{1}{36} \left(39-\sqrt{21}\right) \approx 0.95604, \\
        w_N(0.2,3) &=  \frac{1}{6} \left(\sqrt{21}+6\right) \approx 1.76376, \quad \mbox{and} \\
        w_A(0.2,3) &= \frac{1}{12} \left(33-7 \sqrt{21}\right) \approx 0.0768308.
\end{align*}
In particular, we find that for all $(t,q)$ in Region I we have the following order
\begin{align*}
w_A(t,q) < z=N_-(t,q) < w_N(t,q).
\end{align*}
In particular, since $\lim_{w \to \infty} R_{z,t,q}(w)$ is finite and since the graph of $R_{z,t,q}(w)$ does not cross the diagonal at $w_N(t,q)$ we find that $R_{z,t,q}(w) < w$ for all $w \in (w_A(t,q),w_N(t,q))$.  This implies that the orbit of $z = N_-(t,q)$ under iteration of $R_{z,t,q}(w)$ converges to $w_A(t,q)$, as claimed.   This proves Claim (ii)(a).

\vspace{0.1in}
\noindent
{\bf Region II:}
Consider $(t,q) = (1/6,6)$, which lies in Region II.  For these parameters one finds:
\begin{align*}
z=N_-(1/6,6) &= \frac{1}{320} (-3) \left(\sqrt{33}-111\right) \approx 0.98677, \\
        w_N(1/6,6) &=  \frac{1}{20} \left(\sqrt{33}+9\right) \approx 0.737228, \quad \mbox{and} \\
        w_A(1/6,6) &= \frac{1}{80} \left(69-11 \sqrt{33}\right) \approx 0.0726226.
\end{align*}
In particular, we find that for all $(t,q)$ in Region I we have the following order
\begin{align*}
w_A(t,q) <  w_N(t,q) < z=N_-(t,q).
\end{align*}
since $\lim_{w \to \infty} R_{z,t,q}(w)$ is finite we find that $R_{z,t,q}(w) < w$ for all $w \in (w_N(t,q),\infty)$.  This implies that the orbit of $z = N_-(t,q)$ under iteration of $R_{z,t,q}(w)$ converges to $w_N(t,q)$, as claimed.   

Note that if $t = t_1(q)$ it is clear from the above calculations and continuity
that $z=N_+(t,q)= w_N(t,q)$.   Combined with the analysis of the parameters in Region II, this proves Claim (ii)(b).

\vspace{0.1in}
{\bf Proof of Claim (iii):}
First suppose that $z=N_-(t,q)$ for $(t,q)$ in Region II, as discussed in Case (ii) above.  Since the graph of
$R_{z,t,q}(w)$ does not cross the diagonal at the fixed point $w_N(t,q)$ we have that for all $w \in (w_A(t,q),\infty)$ that $R_{z,t,q}(w) \leq w$ with equality occurring at $w=w_N(t,q)$.  Decreasing the parameter
$z > 0$ decreases the values of $R_{z,t,q}(w)$ and this will eliminate the neutral fixed point $w_N(t,q)$.
Meanwhile, if the decrease of the parameter $z > 0$ is by a sufficiently small amount, 
then the attracting fixed point $w_A(t,q)$ will move continuously to an attracting fixed point $w_A'$ for the perturbed map, and we 
will have that $R_{z,t,q}(w) < w$ for all $w \in (w_A',\infty)$.   This implies that the orbit of
the marked point $w=z$ will now converge to the new attracting fixed point $w_A'$, as claimed in Part (iii) of the Lemma.
\end{proof}

\section{Derivation of the Renormalization Mapping (Proof of Theorem D)}\label{APPENDIX_THMC}

\begin{proof}[Proof of Theorem D]
 For each $0 \leq j \leq q-1$ we define the conditional partition functions of $\Gamma_n$ conditioned on the spin $\sigma(r)$ at the root vertex equaling $j$ as follows
\begin{align*}
Z_n^j&\equiv Z_n^j(z,t):=\sum_ {\sigma\ s.t\ \sigma(r)=j} W_n(\sigma).
\end{align*}
Here $W_n(\sigma):=e^{-\frac{H_n(\sigma)}{T}}$ is the Boltzmann-Gibbs weight of the configuration~$\sigma$.  Notice that
\begin{align}\label{(2.2)}
    \text{$Z_n^j=Z_n^k$ for any $1\leq j,k \leq q-1$.}
\end{align}
To see why this is true, let $\rho$ be a permutation on $\{0,1,\cdots, q-1\}$ which fixes 0. Then $H_n(\sigma) = H_n(\rho \circ \sigma)$, implying the claim. The term $Z^0_n$ is different from $Z^k_n$ for $k=1,\cdots, {q-1}$ because of the term $- h\sum_{i\in V_n}\delta(\sigma(i),0)$ that corresponds to the interaction between the externally applied magnetic field and the spins $\sigma(i)$ in (\ref{(2.1)}). The full partition function then equals to
\[Z_n \equiv
 Z_n(z,t) = Z_n^0+(q-1)Z_n^1.\]

\noindent
{\bf Computing $Z_0^0$ and $Z_0^1$:}
Notice that $\Gamma_n$ is just the root vertex $r$ when $n=0$, so that $H(\sigma)=-h$ when $\sigma(r)=0$ and $H(\sigma)=0$ when $\sigma(r)=1$. Hence, we have 
\begin{align}
   Z_0^0=e^{h/T}=z^{-1} \qquad \mbox{and} \qquad Z_0^1=e^0=1. \label{2.2}
\end{align}

\noindent
{\bf Computing $Z_{n+1}$ in terms of $Z_n^0$ and $Z_n^1$:}
We first compute $Z_{n+1}^0$ and $Z_{n+1}^1$ in terms of $Z_n^0$ and $Z_n^1$. Let $\sigma: V_{n+1}\longrightarrow \{0,\cdots,q-1\}$ be a spin configuration. 
Let $\Lambda_n$ be the graph that can be obtained by joining a single new vertex (say $\alpha$) to the root vertex $r$ of $\Gamma_n$ using a new edge (see Figure \ref{Graph of Gamma}). By $\Lambda_n^2$ let us denote the graph that one can obtain by gluing two $\Lambda_n$ graphs together at the vertex $\alpha$ (see Figure \ref{Graph of two Gammas}). In fact $\Gamma_{n+1}=\Lambda_n^2$. Recall from Equation (\ref{(2.1)}) that $H_n(\sigma)$ denotes the Hamiltonian of the $n^{th}$ level Cayley rooted tree with branching number two for the spin configuration $\sigma$. By $H_{\Lambda_n}(\sigma)$, let us denote the Hamiltonian of $\Lambda_n$ with spin configuration $\sigma$. Let $\mathcal{W}_{n+1}(\sigma)$ be the Boltzmann-Gibbs weight of configuration $\sigma$ for the graph $\Lambda_n$, and let $\mathcal{Z}^{\sigma(\alpha)}_{n+1}$ be the conditional partition function of $\Lambda_n$, supposing the spin at $\alpha$ is $\sigma(\alpha)$.  Here we extend the domain of $\sigma$ from $V_{n+1}$ to $V_{n+1}\cup\{\alpha\}$. Thus, ${\sigma(\alpha)\in \{0,\cdots, q-1\}}$.

\begin{figure}[h]
    \centering
    \includegraphics[width=0.8\linewidth]{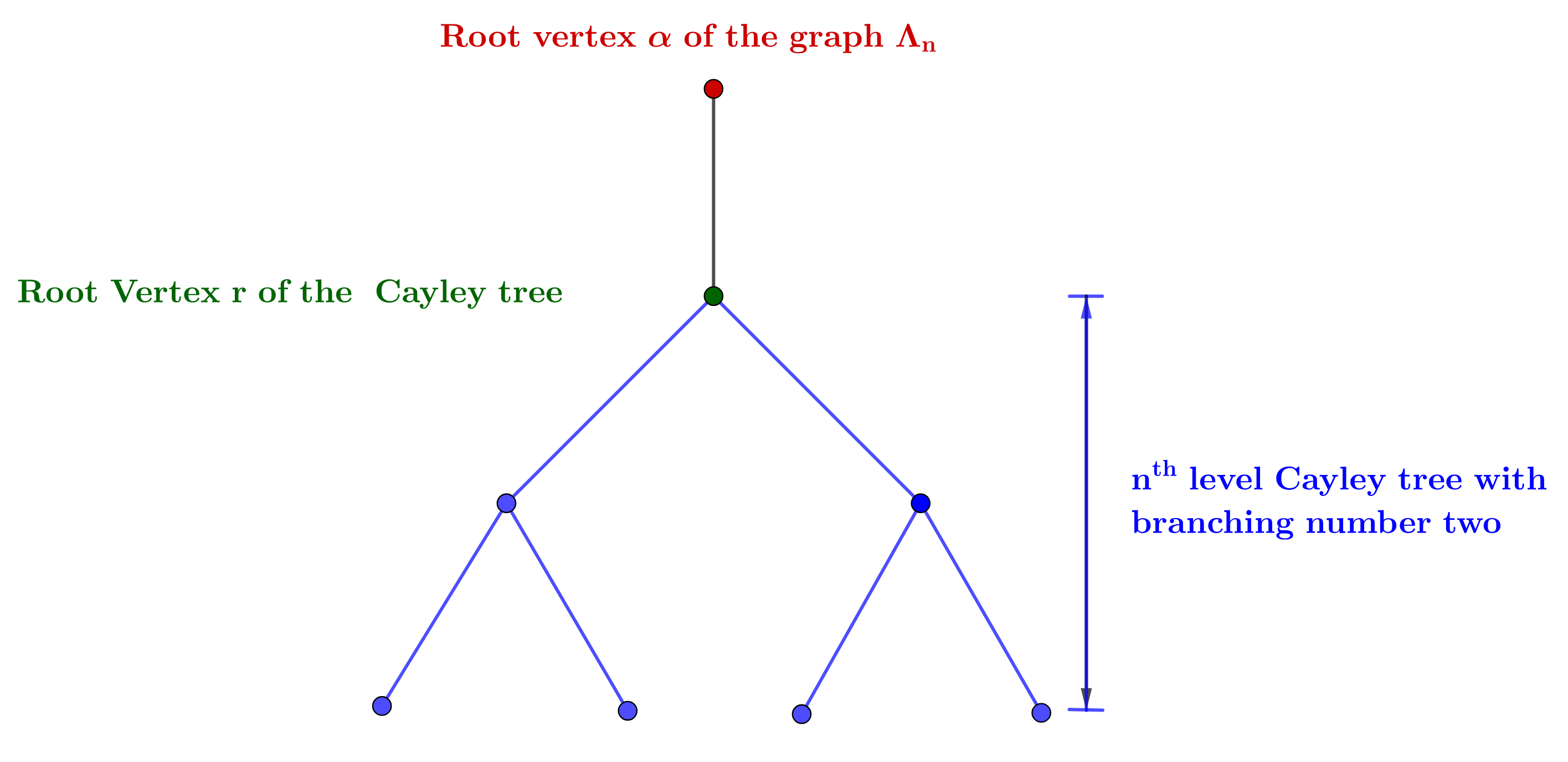}
    \caption{Graph $\Lambda_n$.}
    \label{Graph of Gamma}
 \end{figure}

\begin{figure}[h]
    \centering
    \includegraphics[width=0.7\linewidth]{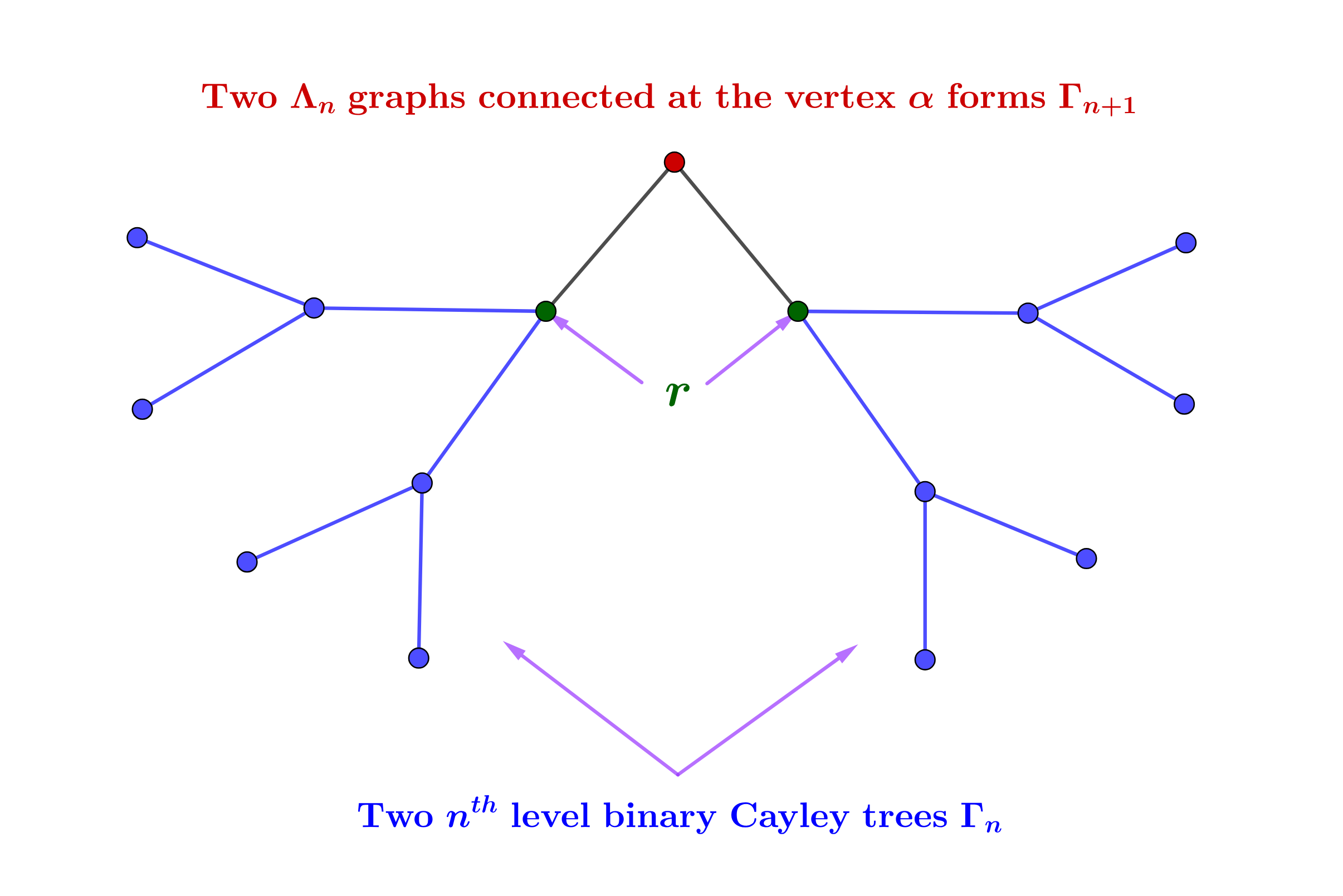}

    \caption{Graph $\Gamma_{n+1}=\Lambda_n^2$: two $\Lambda_n$ connected at their respective vertices $\alpha$.}
    \label{Graph of two Gammas}
\end{figure}

\vspace{0.1in}

\noindent
{\bf Computing $Z_{n+1}^0$ in terms of $Z_n^0$ and $Z_n^1$:}
First, we compute the conditional partition function $\mathcal{Z}^0_{n+1}$ of $\Lambda_n$ in terms of $z,t,Z^0_n$ and $Z^1_n$. Notice that the superscript $0$ on $Z_{n+1}^0$ indicates that we will have $\sigma(\alpha)=0$ throughout this entire subsection. We have two cases.

\vspace{0.1in}
\noindent
{\em Case 1: $\sigma(r)=0$}. 
    Then $H_{\Lambda_n}(\sigma)=H_n(\sigma)-h$, implying that $\sum_{\sigma(r)=0}\mathcal{W}_{n+1}=z^{-1}Z_n^0$.

\vspace{0.1in}
\noindent
{\em Case 2: $\sigma(r)\neq 0$}.
     Then $H_{\Lambda_n}(\sigma)=H_n(\sigma)+J-h$, implying that $\sum_{\sigma(r)\neq 0}\mathcal{W}_{n+1}=(q-1)tz^{-1}Z_n^0$.  Here we have used  Equation (\ref{(2.2)}).
     
\vspace{0.1in}
\noindent
Thus 
\begin{align}
    \mathcal{Z}_{n+1}^0=z^{-1}\big(Z_n^0+(q-1)tZ_n^1 \big).\label{e2.2}
\end{align}
Notice that, when $\sigma(\alpha)=0$, 
\[H_{n+1}=\text{Hamiltonian of}\ \Gamma_n^2 = H_{\Gamma_n}(\sigma)+H_{\Gamma_n}(\sigma)+h.\]
Here the term $``h"$ comes from the Hamiltonian of the graph of a single vertex $\alpha$ (subtracting the energy of the single point $\alpha$). Thus, $Z_{n+1}^0=z(\mathcal{Z}_{n+1}^0)^2$. This together with (\ref{e2.2}) this gives us
\begin{align}
    Z_{n+1}^0 = z^{-1}\big(Z_n^0+(q-1)tZ_n^1\big)^2.\label{e2.3}
\end{align}

\noindent
{\bf Computing $Z_{n+1}^1$ in terms of $Z_n^0$ and $Z_n^1$:}
First, we compute the conditional partition function $\mathcal{Z}^1_{n+1}$ of $\Lambda_n$ in terms of $z,t,Z^0_n$ and $Z^1_n$. Notice that the superscript $1$ on $Z_{n+1}^1$ indicates that we will have $\sigma(\alpha)=1$ throughout this entire subsection. We have three cases.\\

\vspace{0.1in}
\noindent
{\em Case 1: $\sigma(r)=0$}.
    Then $H_{\Lambda_n}(\sigma)=H_n(\sigma)+J$, implying that $\sum_{\sigma(r)=0}\mathcal{W}_{n+1}=tZ_n^0$.

\vspace{0.1in}
\noindent
{\em Case 2: $\sigma(r)=1$}.
    Then $H_{\Lambda_n}(\sigma)=H_n(\sigma)$, implying that $\sum_{\sigma(r)=1}\mathcal{W}_{n+1}=Z_n^1.$

\vspace{0.1in}
\noindent
{\em Case 3: $\sigma(r)\neq 0,1$}.
    Then $H_{\Lambda_n}(\sigma)=H_n(\sigma)+J$, implying that $\sum_{\sigma(r)\notin \{0,1\}}\mathcal{W}_{n+1}=(q-2)tZ_n^1$.  Here we have used  Equation (\ref{(2.2)}).

\vspace{0,1in}
\noindent Thus 
\begin{align}
    \mathcal{Z}_{n+1}^1=tZ_n^0 + Z_n^1 + (q-2)tZ_n^1.\label{e2.4}
\end{align}
Notice that, when $\sigma(\alpha)=1$, 
\[H_{n+1}=\text{Hamiltonian of}\ \Gamma_n^2 = H_{\Gamma_n}(\sigma)+H_{\Gamma_n}(\sigma).\]
Thus $Z_{n+1}^1=\big(\mathcal{Z}_{n+1}^1\big)^2$. This together with Equation (\ref{e2.4}) gives us
\begin{align}
    Z^1_{n+1}=\big(tZ_n^0 + Z_n^1 + (q-2)tZ_n^1\big)^2.\label{e2.5}
\end{align}

Therefore, by (\ref{e2.3}) and (\ref{e2.5}) we have a formula for the full partition function of the rooted Cayley tree with branching number two at level $n+1$:
\begin{align}
    \nonumber Z_{n+1} &:= Z_{n+1}^0 + (q-1)Z_{n+1}^1\\
    &= z^{-1}\big(Z_n^0+(q-1)tZ_n^1 \big)^2 + (q-1)\Big(tZ_n^0+Z_n^1+(q-2)tZ_n^1 \Big)^2. \label{2.15}
    \end{align}
 We are interested in the zeros of $Z_{n+1}$. Notice that $Z_{n+1}=0$ when 
 \begin{align}
      \frac{Z_{n+1}^1}{Z_{n+1}^0}=z\frac{\big(tZ_n^0+Z_n^1+(q-2)tZ_n^1 \big)^2}{\big(Z_n^0+(q-1)tZ_n^1 \big)^2}
      = z\Bigg[\frac{tZ_n^0+Z_n^1+(q-2)tZ_n^1 }{Z_n^0+(q-1)tZ_n^1 }\Bigg]^2
      =\frac{1}{1-q}. \label{2.7}
 \end{align}
Let $w_n:=\frac{Z_{n}^1}{Z_{n}^0}.$ Then by (\ref{2.7}),
\begin{align}
    \nonumber w_{n+1} &=  z\Bigg[\frac{tZ_n^0+Z_n^1+(q-2)tZ_n^1 }{Z_n^0+(q-1)tZ_n^1 }\Bigg]^2
    =  z\Bigg[\frac{t+\frac{Z_n^1}{Z_n^0}+(q-2)t\frac{Z_n^1}{Z_n^0} }{1+(q-1)t\frac{Z_n^1}{Z_n^0} }\Bigg]^2
    = z\Bigg[\frac{t+w_n+(q-2)tw_n }{1+(q-1)tw_n }\Bigg]^2.
\end{align}
 \noindent Thus if we define 
 \begin{align}
     R_{z,t,q}(w):= z\Bigg[\frac{t+w+(q-2)tw }{1+(q-1)tw }\Bigg]^2, \label{2.18}
 \end{align}
 then \[w_{n} = R_{z,t,q}^{n}(w_0).\]
Here, the superscript $``n"$ indicates that we compose $R_{z,t,q}(w)$ with itself $n$ times (with respect to the variable $w$ while fixing parameters $z, t$ and $q$). In other words, it denotes the $n^{th}$ iterate of the function $R_{z,t,q}$.

By definition and (\ref{2.2}) we have $w_0=\frac{Z_0^1}{Z^0_0} = \frac{1}{z^{-1}} = z$. So,  $w_{n} = R_{z,t,q}^{n}(z)$.
Therefore, by ($\ref{2.15}$) and ($\ref{2.18}$) the Lee-Yang zeros of the $q$-state Potts model on the $n^{th}$ level Cayley tree with branching number two are the solutions to the equation
\begin{align}
    R_{z,t,q}^{n} (z) = \frac{1}{1-q}.\label{2.20}
\end{align}
This finishes the proof of Part (i) of Theorem D.

\vspace{0.1in}
\noindent
{\bf Case of the unrooted Cayley tree:}
We now explain how to prove Part (ii) of Theorem D about the unrooted Cayley $\hat \Gamma_n$.
Denote by $c$ the central vertex of $\hat \Gamma_n$, i.e. the vertex at distance $n$
from the leaves of $\hat \Gamma_n$.   For any $0 \leq j \leq q-1$
denote by $\hat Z_n^j$ the conditional partition function for $\hat \Gamma_n$ conditioned on $\sigma(c) = j$.   Like in the proof for the rooted tree, we again  have $\hat Z_n^j = \hat Z_n^k$ for any $1 \leq j,k \leq q-1$.   

One obtains $\hat \Gamma_n$ by taking three copies
of the rooted tree $\Gamma_{n-1}$ and attaching each of their root vertices by an
edge to the central vertex $c$.   In much the same way as we handled the rooted
tree above, one can prove that:
\begin{align*}
\hat{Z}^0_{n} =  z^{-1}\big(Z_{n-1}^0+(q-1)tZ_{n-1}^1\big)^3 \qquad \mbox{and} \qquad  \hat{Z}^1_{n} = \big(tZ_{n-1}^0 + (1+ (q-2)t)Z_{n-1}^1\big)^3.
\end{align*}
Therefore, $\hat{Z}_n = \hat{Z}^0_{n} + (q-1) \hat{Z}^1_{n} = 0$ if and only if
\begin{align*}
\frac{\hat{Z}^1_{n}}{\hat{Z}^0_{n}} = z \frac{\big(tZ_{n-1}^0 + (1+ (q-2)t)Z_{n-1}^1\big)^3}{\big(Z_{n-1}^0+(q-1)tZ_{n-1}^1\big)^3}  = \hat{R}_{z,t,q}(w_{n-1})  =  \left(\hat{R}_{z,t,q} \circ R^{(n-1)}_{z,t,q}\right)(z) = \frac{1}{1-q}.
\end{align*}
This finishes the proof of Part (ii) of Theorem D.
\end{proof}

\vspace{0.1in}

\noindent
{\bf Disclosures and declarations:}

\noindent
As required by current rules for articles submitted to Springer-Nature journals, the authors declare that:
\begin{enumerate}

\item The authors have no financial or non-financial or competing interests that are relevant to this article. The authors have stated their grant support in the Acknowledgments.

\item Data accessibility: this article is theoretical and does not contain any experimental data. Calculational data supporting the conclusions of the article are included in the text of the article.
\end{enumerate}

\bibliographystyle{plain}
\bibliography{references2}

\end{document}